\documentclass[a4paper,onecolumn,11pt,accepted=2017-05-09]{quantumarticle}
\pdfoutput=1
\usepackage[utf8]{inputenc}
\usepackage[english]{babel}
\usepackage[T1]{fontenc}
\usepackage{amsmath}
\usepackage{amsthm}
\usepackage{amsfonts}
\usepackage{hyperref}
\usepackage{mathrsfs}

\usepackage{tikz}
\usepackage{lipsum}
\newtheorem{theorem}{Theorem}
\usepackage{graphicx}
\usepackage{amssymb}
\usepackage{algorithmicx, algorithm}
\usepackage{algpseudocode}
\usepackage{mathtools}
\usepackage{physics}
\usepackage{qcircuit}
\usepackage{xcolor,soul}
\setstcolor{red}
\usepackage{subcaption,xspace}
\usepackage{array}
\usepackage{thmtools} 
\usepackage{thm-restate}
\usepackage{accents}
\usepackage{cleveref}
\allowdisplaybreaks

\newtheorem{problem}{Problem}
\newtheorem{algo}{Algorithm}

\newcommand{\hdq}{\texttt{HD$_q$}\xspace}
\newcommand{\cmp}{\texttt{CMP}\xspace}
\newcommand{\pmax}{p_{max}\xspace}

\newcommand{\D}{\mathcal{D}\xspace}
\newcommand{\Ot}{\Tilde{O}}

\newcommand{\prob}{\textsc{ProFil}\xspace}
\newcommand{\amprob}{\textsc{AmpFil}\xspace}

\newcommand{\ampest}{\textsc{AmpEst}\xspace}
\newcommand{\trueampest}{\textsc{TrueAmpEst}\xspace}
\newcommand{\eqae}{\textsc{EQAmpEst}\xspace}

\newcommand{\kdist}{$k$-\textsc{Distinctness}\xspace}
\newcommand{\gkd}{\textsc{Gapped $k$-Distinctness}\xspace}

\newcommand{\Finf}{$\textsc{F}_\infty$\xspace}

\newcommand{\countdec}{\textsc{CountDecision}\xspace}
\newcommand{\ED}{\textsc{ElementDistinctness}}

\newcommand{\deltaerr}{\log\tfrac{1}{\delta\tau}}

\newcommand{\pmset}{S}
\newcommand{\fmax}{\hat{f}_{max}\xspace}

\renewcommand{\O}{\mathcal{O}}
\newcommand{\borcl}{\hat{\mathcal{O}}}
\newcommand{\hata}{\hat{A}}

\newcommand{\iden}{\mathbb{I}}
\newcommand{\mbc}[1]{\mathbf{C}^{(#1)}}

\newcommand{\gamy}{\gamma^y}

\newcolumntype{C}[1]{>{\centering\let\newline\\\arraybackslash\hspace{0pt}}m{#1}}
{\bfseries}{\itshape}

\algnewcommand\algorithmicon{\textbf{on}}
\algdef{SE}[ON]{On}{EndOn}[1]{\algorithmicon\ #1\ \algorithmicdo}{\algorithmicend\ \algorithmicloop}
\newcommand{\alg}{\mathcal{A}}

\newcommand{\mdistampest}{MDistAmpEst\xspace}
\newcommand{\mdaealgo}{\textsc{MDistAEAlgo}\xspace}
\newcommand{\ampfilorcl}{\textsc{AmpFilBOrcl}\xspace}
\newcommand{\dist}{\mathcal{D}_p}
\newcommand{\probfilorcl}{\textsc{ProbFilBOrcl}\xspace}

\newcommand{\aealgo}{\textsc{QAEAlgo}\xspace}

\newtheorem{lemma}[theorem]{Lemma}
\newtheorem{claim}[theorem]{Claim}
\newtheorem{corollary}[theorem]{Corollary}
\newtheorem{proposition}[theorem]{Proposition}

\begin{document}

\title{Few Quantum Algorithms on Amplitude Distribution}

\author{Debajyoti Bera}
\email{dbera@iiitd.ac.in}
\affiliation{Department of Computer Science, IIIT-D, New Delhi, India.}
\author{Tharrmashastha SAPV}
\email{tharrmashasthav@iiitd.ac.in}
\affiliation{Department of Computer Science, IIIT-D, New Delhi, India.}
\maketitle

\begin{abstract}
  Amplitude filtering is concerned with identifying basis-states in a superposition whose amplitudes are  greater than a specified threshold; probability filtering is defined analogously for probabilities. Given the scarcity of qubits, the focus of this work is to design log-space algorithms for them.

Both algorithms follow a similar pattern of estimating the amplitude (or, probability for the latter problem) of each state, in superposition, then comparing each estimate against the threshold for setting up a flag qubit upon success, finally followed by amplitude amplification of states in which the flag is set. We show how to implement each step using very few qubits by designing three subroutines. Our first algorithm performs amplitude amplification even when the ``good state'' operator has a small probability of being incorrect --- here we improve upon the space complexity of the previously known algorithms. Our second algorithm performs ``true amplitude estimation'' in roughly the same complexity as that of ``amplitude estimation'', which actually estimates a probability instead of an amplitude. Our third algorithm is for performing amplitude estimation in parallel (superposition) which is difficult when each estimation branch involves different oracles.

As an immediate reward, we observed that the above algorithms for the filtering problems directly improved the upper bounds on the space-bounded query complexity of problems such as non-linearity estimation of Boolean functions and $k$-distinctness.
\end{abstract}

\section{Introduction}

\newcommand{\C}{\mathcal{C}}
\newcommand{\A}{\mathcal{A}}

A quantum circuit is always associated with a distribution, say $\D$, over the observed outcomes~\footnote{We assume measurement in the standard basis in this paper, however, it should not be difficult to extend our algorithms to accommodate measurements in another basis.} that can, in principle, encode complex information. Given a threshold $\tau$, and a blackbox to run the circuit, it may be useful to know if there is an outcome with a probability of at least $\tau$. We call this problem \textsc{Probability Filtering} (denoted \prob). We also introduce \textsc{Amplitude Filtering} (denoted \amprob) that determines if the absolute value of the amplitude of any basis state is above a given threshold; even though this problem appears equivalent to \prob, an annoying difference crawls in if we allow absolute or relative errors with respect to the threshold. We are not aware of prior algorithms for these problems, as stated in their most general form. The most interesting takeaway from this work are $\Ot(1)$-qubit\footnote{$\tilde{O}$ hides $\log$ factors.} algorithms for the \prob and the \amprob problems whose query complexities, measured as the number of calls to the circuit, are independent of the domain size of $\D$.

The framework offered by these problems supports interesting tasks. For example, a binary search over $\tau$ (tweaked to handle the above annoyance) can be a way to compute the largest probability among all the outcomes and can be used to find the modal outcome. We have observed that several combinatorial problems can be reduced to finding the mode of a distribution or identifying if the mode is greater than something. For instance, the \kdist problem generalizes the well-studied element-distinctness problem: whether an array has at least $k$ repetitions of any element. Consider the distribution over the domain of the array elements. If any element appears at least $k$ times, then its mode will be at least $k/N$ ($N$ denoting the size of the array), and vice versa. Our algorithm for probability filtering can be used to design an algorithm for \kdist that makes an optimal number of queries (up to logarithmic factors) when $k=\Omega(N)$, and that too using $\tilde{O}(1)$ qubits. Previous quantum algorithms for large $k$ have an exponential query complexity and require a larger number of qubits~\cite{Ambainis2007QuantumDistinctness}.
We hope that the frameworks of \prob and \amprob would be useful for designing more quantum algorithms.

When space is not a constraint, query complexity of a discrete problem with an $n$-sized domain is $O(n)$ which is achievable by querying and caching the entire input at the beginning. However, this is not feasible when space is limited. In contrast, our algorithms are allowed only constant many logarithmic-sized registers. 
To restrict the number of qubits to $\tilde{O}(1)$ we end up using super-linear queries for many problems.

\subsection{Summary of Results}

The algorithms for \prob and \amprob are built on three different algorithms.

\paragraph*{Biased Amplitude Amplification in Log-space~(\cref{sec:aa-biased-oracle})}The \prob problem can be solved using an intuitively simple idea of doing amplitude estimation for each outcome in superposition, using the threshold as a marking function, and then doing amplitude amplification with respect to the marking. Doing this while ensuring that the errors inherent in the estimation step do not increase significantly in the amplification step can actually be reduced to the problem of biased-oracle amplitude amplification. 

In the latter problem, the oracle to mark ``good'' states is allowed to err with some probability $1-p$. Hoyer et al.~\cite{biasedoraclehoyer} studied this problem earlier for $p=9/10$. 
They proposed an algorithm that uses $O(\sqrt{1/\lambda})$ queries to obtain a marked element with high probability where $\lambda$ is the probability of obtaining any marked element out of a total $N$ elements.
This algorithm performs an ``error reduction step'' after each iteration of amplitude amplification. 
This step uses one new qubit for reducing the error; however, since that additional qubit becomes completely entangled with the other qubits, a fresh qubit is required for each error reduction step. Thus, the number of qubits required would grow as much as $O(\sqrt{1/\lambda})$ in the worst case.

To reduce the qubit footprint of our algorithms, we designed our own algorithm for biased oracle amplitude amplification for arbitrary $p>1/2$ that uses $O\Big(\frac{p}{(p-\frac{1}{2})^2\sqrt{\lambda}}\log(\frac{1}{\lambda\delta})\Big)$ queries and just $\log(N)+O\Big(\frac{2p}{(p-1/2)^2}\log(1/\lambda\delta)\Big)$ qubits (which is $\Ot(1)$ for $p>2/3$ and $\lambda=1/poly(N)$).

\paragraph*{True Amplitude Estimation (\cref{sec:true-amp-est})}
The algorithm for biased amplitude amplification can be used to solve a promise version of \prob almost optimally, and the latter could be used to solve \amprob, but only sub-optimally. The reason is that this approach performs amplitude estimation for each outcome; however, amplitude estimation would actually estimate the probabilities of the outcomes. Suppose we represent a state in the following manner: $\ket{\psi} = \alpha \ket{Good} + \beta \ket{Bad}$, and we want to estimate $|\alpha|$ within an small additive error $\epsilon$. The standard amplitude estimation technique can be used to estimate $|\alpha|^2$ with error $\epsilon^2$, which can be used to indirectly obtain an estimate of $|\alpha|$ with error $\epsilon$. However, this requires $O(1/\epsilon^2)$ queries.


To speed up this approach, we present an algorithm for {\em true amplitude estimation} (\trueampest)~\footnote{We use ``true'' to differentiate it from the well-known ``amplitude estimation'' algorithm (\ampest), which as explained above, is a misnomer.} that estimates the amplitude (rather, its norm) with a query complexity of $O(\frac{1}{\epsilon})$ queries, using a similar number of qubits as that required for amplitude estimation.
The algorithm uses the native amplitude estimation algorithm, denoted \ampest, and a variant of the Hadamard test as a subroutine.~\footnote{The algorithm is an extension of the amplitude estimation algorithm, and could have been used earlier. Since we did not find it anywhere, we are including it for the sake of completeness.}

\paragraph*{Multidistribution Amplitude Estimation in Log-space (\cref{sec:mdim-amp-est})}

The final challenge comes in trying to run \trueampest algorithm in superposition to estimate the amplitudes of every $\ket{x}$ in a state $\sum_x \alpha_x \ket{x}$. This boils down to running \ampest in superposition. It is quite easy to parallelly estimate the probabilities associated with different basis states in a fixed superposition, e.g., estimating $|\alpha_x|^2$ for all $x$ in a state $\sum_x \alpha_x \ket{x}$. In fact, we would be following this approach for \prob. Apeldoorn et al.\ had studied a similar problem named {\em multidimensional amplitude estimation}~\cite{van2021quantum}.
Their objective was to estimate each $|\alpha_x|^2$ separately with each estimate in a separate register.
They provided an $O(1/\epsilon)$-query algorithm but its space requirement is understandably $\Omega(n)$, in fact, it requires $O(n/\epsilon)$ qubits of space. 

Now consider a different parallel estimation problem, that of  parallelly estimating the probability associated with a fixed basis state in different superpositions, e.g., estimating $|\alpha_{y,0}|^2$ for all $y$ in a superposition of $\sum_x \alpha_{y,x} \ket{x}$ over several $y$'s.
Recall that \ampest uses two operators as a black-box -- one for generating the state on which estimation should be performed and another for marking the ``good states''. Technically speaking, we face the difficult problem of parallelizing \ampest where a different state generation operator is employed in each branch of a superposition~\footnote{This was not a problem for the above algorithm for \prob since the state generation oracle was common and the marking oracles were implemented using the register that was in superposition.}.
We name this task as 
{\em multidistribution amplitude estimation} (\mdistampest).



We designed a quantum algorithm for the above problem 
with the same query complexity as a single estimation procedure viz., $O(1/\epsilon)$ queries and $O(\log(n/\epsilon))$ space. The query complexity of a straightforward implementation has an overhead of $n$, the number of branches in the superposition, against a single amplitude estimation. We managed to keep the query complexity low by pushing controlled operations deep inside the algorithm to identify and eliminate redundant calls to the oracle.

\paragraph*{Algorithms for \prob and \amprob (\cref{sec:profil_algo})}
Using our biased amplitude amplification algorithm with the amplitude estimation, we could solve a promise version of \prob in which either there is some outcome with probability more than the threshold $\tau$ or the probability of every outcome is less than $\tau-\epsilon$. Our query complexity is $\Ot(\tfrac{1}{\epsilon\sqrt{\tau}})$.

Our next objective was to design a $\Ot(\tfrac{1}{\epsilon \tau})$-query algorithm for a similar promise version of the \amprob problem. Here, the algorithm for \prob could be directly employed, using threshold as $\tau^2$; however, the query complexity would scale as $\tfrac{1}{\epsilon^2}$ instead of the desired $\tfrac{1}{\epsilon}$. This situation arises since the amplitude estimation step actually estimates the probability and not the amplitude, so the estimation inaccuracy too scales quadratically with respect to the latter. We use \trueampest to work around this hurdle.
Thus, we obtain an algorithm for the \amprob problem that follows the same overall approach as that of \prob but with vastly different subroutines: estimating the amplitude of all the states in superposition (using \trueampest and \mdistampest), marking the states whose amplitude's modulus is larger than the given threshold in parallel, and then amplifying the probability of finding a marked state (using biased amplitude amplification).
We show that the \amprob problem can be solved in $\tilde{O}(\frac{1}{\epsilon \tau})$ queries.

For, the \prob problem, we show that our algorithm is tight with respect to the parameters $\epsilon$ and $\tau$ individually, i.e, we show a lower bound of $\Omega(\frac{1}{\epsilon}+\frac{1}{\sqrt{\tau}})$ queries.
Further, we show an almost tight lower bound of $\Omega(\frac{1}{\epsilon}+\frac{1}{\tau})$ queries for the \amprob problem. 
Both the lower bounds use standard approaches like the adversary method~\cite{adversarymethod} and reduction from a counting problem~\cite{Nayak1999QuantumStatistics}. The details on the lower bounds is presented in Section~\ref{sec:highdist-lb}.


\paragraph*{Applications of \prob and \amprob}

The results in this work can be used to design low-space, often $\Ot(1)$, algorithms for several problems which have received recent attention. These problems can now be solved using a logarithmic number of qubits --- often exponentially less compared to the existing approaches, and have a better query complexity, thus leading to better space-time complexities.
The reductions are mostly straightforward and are omitted due to space constraints, but the implications are interesting, as discussed below.
    \begin{itemize}
        \item Our approach for \kdist makes an optimal number of queries (up to logarithmic factors) when $k=\Omega(n)$, and that too using $\tilde{O}(1)$ qubits (see section~\ref{sec:kdist}). Previous quantum algorithms for large $k$ have an exponential query complexity in limited space or require a polynomial number of qubits~\cite{Ambainis2007QuantumDistinctness}.
        \item Our algorithm for \prob can be used to identify the presence of high-frequency items in an array (those above a given threshold --- this problem is also known as ``heavy hitters'') using $\tilde{O}(\log\tfrac{1}{\epsilon})$ qubits; it also generates a superposition of such items along with estimates of their frequencies. The best algorithms for identifying heavy hitters in low space classical algorithms are of streaming nature but require $\tilde{O}(\tfrac{1}{\epsilon})$ space~\cite{CMSketch}. Here $\epsilon\in(0, 1]$ indicates the inaccuracy in frequency estimation.
        \item Our \prob and \amprob algorithms can be used to binary search for the largest threshold, which essentially yield the largest probability and the largest amplitude, respectively.
        \item Valiant and Valiant showed that $\tilde{O}(\tfrac{m}{\epsilon^2})$ samples of an $m$-valued array are sufficient to classically estimate common statistical properties of the distribution of values in the array~\cite{ValiantValiant}. Recently it was shown that fewer samples of the order of $\tilde{O}(\tfrac{1}{g^2})$ could be used if we want to identify the item with the largest probability (denoted $p_{\max}$)~\cite{Dutta2010ModeDistributions}; here $g$ denotes the gap between $p_{\max}$ and the second largest probability which is always less than $\pmax$. A binary search using \prob makes only $\tilde{O}(\tfrac{1}{g \sqrt{p_{\max}}})$ queries and is able to locate such an item.
        \item The non-linearity of a Boolean function can also be calculated in terms of the amplitude with the largest norm among the output superposition of the Deutsch-Jozsa circuit. An algorithm based on this connection was recently proposed that employed binary search using \prob to estimate the non-linearity of any Boolean function with additive accuracy $\lambda$ using $\Ot(1)$ qubits and $\Ot(\tfrac{1}{\lambda^2 \fmax})$ queries~\cite{Bera2021QuantumEstimation}; here $\fmax$ denotes the largest absolute value of any Walsh coefficient of the function, also the largest amplitude in the output state of the Deutsch-Jozsa circuit. Now, we can use $\amprob$ instead of \prob to do the same but using only $\Ot(\tfrac{1}{\lambda\fmax})$ queries (this also involves doing a binary search over all $|\hat{f}(x)|$). It should be noted in this context that the best known lower bound for non-linearity estimation is $\Omega(\tfrac{1}{\lambda})$ (Appendix~\ref{appendix:non-lin-lb}).
    \end{itemize}
\section{Background: Quantum amplitude estimation (QAE)}
\label{sec:background_qae}
Consider a quantum circuit $A$ on $n$ qubits whose final state is $\ket{\psi}$ on input $\ket{0^n}$. Let $\ket{x}$ be some basis state (in the standard basis --- this can be easily generalized to any arbitrary basis).
Given an accuracy parameter $\epsilon\in (0,1)$, the amplitude estimation problem is to estimate the probability $p$ of observing  $\ket{x}$ upon measuring $\ket{\psi}$ in the standard basis, up to an additive accuracy $\epsilon$.
Brassard et al.~\cite{brassard2002quantum} proposed a quantum amplitude estimation circuit, which we call $\aealgo_A$, that acts on two registers of size $n$ and $m$ qubits and makes $2^m-1$ calls to controlled-$A$ to output an estimate $\tilde{p} \in [0,1]$ of $p$ that behaves as follows.
\begin{theorem}\label{thm:amp_est}
    The amplitude estimation algorithm returns an estimate $\tilde{p}$ that has a confidence
    interval $|p-\tilde{p}| \le 2\pi k \frac{\sqrt{p(1-p)}}{2^m} +
    \pi^2 \frac{k^2}{2^{2m}}$ with probability at least $\frac{8}{\pi^2}$ if
    $k=1$ and with probability at least $1-\frac{1}{2(k-1)}$ if $k \ge 2$. It uses exactly $2^m-1$ evaluations of the oracle. If
    $p=0$ or 1 then $\tilde{p}=p$ with certainty.
\end{theorem}

The following corollary is obtained from the above theorem by setting $k=1$ and $m=q+3$.
\begin{corollary}\label{cor:amp_est_our_form}
    The amplitude estimation algorithm returns an estimate $\tilde{p}$ that has a confidence
    interval $|p-\tilde{p}| \le \frac{1}{2^q}$ with probability at least $\frac{8}{\pi^2}$ using $q+3$ qubits and $2^{q+3}-1$ queries. If
    $p=0$ or 1 then $\tilde{p}=p$ with certainty.

    Setting $\frac{1}{2^q}=\epsilon$, the amplitude estimation algorithm returns an estimate $\tilde{p}$ that has a confidence
    interval $|p-\tilde{p}| \le \epsilon$ with probability at least $\frac{8}{\pi^2}$ using $O\Big(\log(\frac{1}{\epsilon})\Big)$ qubits and $O\big(\frac{1}{\epsilon}\big)$ queries.
\end{corollary}


We use the subscript in $\aealgo_A$ to remind the reader that the circuit for amplitude estimation depends on the algorithm $A$ that generates the state $\ket{\psi}$ from $\ket{0^n}$.

Now, let $p_x$ be the probability of obtaining the basis state $\ket{x}$ on measuring the state $\ket{\psi}$. The amplitude estimation circuit referred to above uses an oracle, denoted $O_x$ to mark the ``good state'' $\ket{x}$, and involves measuring the output of the $\aealgo_A$ circuit in the standard basis; actually, it suffices to only measure the second register. We can summarise the behaviour of the $\aealgo_A$ circuit (without the final measurement) in the following lemma.
\begin{lemma}
    Given an oracle $O_x$ that marks $\ket{x}$ in some state $\ket{\psi}$ and the algorithm $A$ that acts as $A\ket{0^n}=\ket{\psi}$, $\aealgo$ on an input state $\ket{00\ldots 0}\ket{0^m}$ generates the following state.
    $$\aealgo_{A,O_x} \ket{00\ldots 0}\ket{0^m} \xrightarrow{} \beta_{x,s}\ket{\psi}\ket{\hat{p_x}} + \beta_{x,\overline{s}}\ket{\psi}\ket{E_x}$$
    Here, $|\beta_{x,s}|^2$, the probability of obtaining the good estimate, is at least $\frac{8}{\pi^2}$, and $\ket{\hat{p}_x}$ is an $m$-qubit normalized state of the form $\ket{\hat{p}_x} = \gamma_{+}\ket{\hat{p}_{x,+}} + \gamma_{-}\ket{\hat{p}_{x,-}}$ such that both $\sin^2(\pi\frac{\hat{p}_{x,+}}{2^m})$ and $\sin^2(\pi\frac{\hat{p}_{x,-}}{2^m})$ approximate $p_x$ up to $m-3$ bits of accuracy. Further, $\ket{E_x}$ is an $m$-qubit error state (normalized) such that any basis state in $\ket{E_x}$ corresponds to a bad estimate, i.e., we can write 
    $\displaystyle\ket{E_x} = \sum_{\substack{t \in \{0,1\}^m \\ t \not\in \{\hat{p}_{x,+}, \hat{p}_{x,-} \}}}\gamma_{t,x}\ket{t}$ in which $|\sin^2\left(\pi\tfrac{t}{2^m}\right) - p_x| > \tfrac{1}{2^{m-3}}$ for any such $t$.
\end{lemma}

In an alternate setting where the oracle $O_x$ is not provided, $\aealgo_A$ can still be performed if the basis state $\ket{x}$ is provided, say, in a different register. One can construct a quantum circuit, say $EQ$, that acts on basis states as $\ket{x}\ket{y} \mapsto (-1)^{\delta_{x,y}} \ket{x} \ket{y}$. Now perform $\aealgo_A$ in which we replace all calls to $O_x$ by $EQ$ whose first input is set to $\ket{x}$ from the new register.
We name this circuit as $\eqae_A$ that implements the following operation.

$$\eqae_A\big(\ket{x}\ket{00\ldots 0}\ket{0^m}\big) \xrightarrow{} \ket{x}\big(\beta_{x,s}\ket{\psi}\ket{\hat{p}_x} + \beta_{x,\overline{s}}\ket{\psi}\ket{E_x}\big)$$
Further, since $\eqae_A$ is a quantum circuit, we could replace the state $\ket{x}$ by any superposition $\sum_x\alpha_x\ket{x}$. We would be using $\eqae_A$ in this mode in this work.

    $$\eqae_A\Big(\sum_x\alpha_x\ket{x}\ket{00\ldots 0}\ket{0^m}\Big) \xrightarrow{} \sum_x\alpha_x\beta_{x,s}\ket{x}\ket{\psi}\ket{\hat{p}_x} + \sum_x\alpha_x\beta_{x,\overline{s}}\ket{x}\ket{\psi}\ket{E_x}.$$
Let $p_x$ denote the probability of observing the basis state $\ket{x}$ when the state $\ket{\psi}$ is measured.
Notice that on measuring the first and the third registers of the output, with probability $|\alpha_x\beta_{x,s}|^2 \ge \frac{8}{\pi^2}|\alpha_x|^2$ we would obtain as measurement outcome a pair $\ket{x}\ket{\hat{p}_x}$ where $\sin^2(\pi\frac{\hat{p}_x}{2^m}) = \tilde{p}_x$ is within $\pm \frac{1}{2^{m-3}}$ of $p_x$.
Observe in this setting that the subroutine essentially estimates the probabilities $p_x$ corresponding to {\em all} the basis states $\ket{x}$ according to the distribution implicit in the superposition. 
This shows how amplitude estimation can be parallelized to identify all the probabilities in a {\em single} distribution. In a later section, we discuss another approach to parallelize amplitude estimation across multiple distributions.

\section{True Amplitude Estimation}
\label{sec:true-amp-est}


Let's recall amplitude estimation. We are given access to a state preparation oracle $A$ such that $A\ket{0^n}=\sum_{x\in\{0,1\}^n} \alpha_x\ket{x}$, denoted $\ket{\psi}$, and we are also given a ``good state'' $\ket{y}$. With probability at least $8/\pi^2$, \ampest outputs an estimate $\Tilde{\alpha}_y$ of the probability $\alpha_y$ with $\epsilon$ bits of accuracy using $O(\frac{1}{\epsilon})$ calls to $A$.
Notice that contrary to the name, the algorithm does not estimate the `amplitude' of $\ket{y}$ up to $\epsilon$ bits of accuracy; 
rather, what it estimates is its `probability' $p_y=|\alpha_y|^2$.


If one wishes to estimate the absolute value of the amplitude of $\ket{y}$ up to $\epsilon$ accuracy, one could use \ampest to estimate $p_y$ up to $\epsilon^2$ accuracy as $\Tilde{p}_y$ and output $\sqrt{\Tilde{p}_y}$.
It can be easily seen that $\sqrt{\Tilde{p}_y}$ is an $\epsilon$-estimate of $|\alpha_y|$.
However, now the query complexity worsens to $O(1/\epsilon^2)$.
We circumvent this issue and obtain an algorithm that performs an $\epsilon$-estimation of $|\alpha_y| = \sqrt{\Re(\alpha_y)^2 + \Im(\alpha_y)^2}$ in $O(1/\epsilon)$ by using a modified Hadamard test to separately estimate the real and the imaginary parts of $\alpha_y$.
Although the modified Hadamard test is well-known, to the best of our knowledge we could not find it in the form presented here.
We first explain the modified Hadamard test which is one well-known method for estimating the inner product of two states.

Say, we have two algorithms $A_{\psi}$ and $A_{\phi}$ that generate the states $A_{\psi}\ket{0^n}=\ket{\psi}$ and $A_{\phi}\ket{0^n}=\ket{\phi}$, respectively, and we want to produce a state $\ket{0}\ket{\xi_0} + \ket{1}\ket{\xi_1}$ such that the probability (say $p_0$) of observing the first register to be in the state $\ket{0}$ is linearly related to $|\braket{\psi}{\phi}|$. Though swap-test is commonly used for this purpose, there the probability is proportional to $|\braket{\psi}{\phi}|^2$; this subtle difference becomes a bottleneck if we are trying to use amplitude estimation to estimate $p_0$ with additive accuracy, say $\epsilon$. We show that the Hadamard test can do the estimation using $O(1/\epsilon)$ queries to the algorithms whereas it would be $O(1/\epsilon^2)$ if we use the swap test.

The Hadamard test circuit denoted $HT_{A_{\psi}, A_{\phi}}$ requires one additional qubit, initialized as $\ket{0}$ on which the $H$-gate is first applied. Then, we apply a conditional gate controlled by the above qubit that applies $A_{\psi}$ to the second register, initialized to $\ket{0^n}$, if the first register is in the state $\ket{0}$, and applies $A_{\phi}$ if the first register is in the state $\ket{1}$. Finally, the $H$-gate is again applied to the first register. The probability of observing the first register in the state $\ket{0}$ is $\frac{1}{2}\big(1+Re\big(|\bra{\psi}\ket{\phi}|\big)\big)$.

Thus, to obtain $Re\big(|\bra{\psi}\ket{\phi}|\big)$ with $\epsilon$ accuracy, it suffices to estimate $\frac{1}{2}\big(1+Re\big(|\bra{\psi}\ket{\phi}|\big)\big)$ with $\epsilon/2$ accuracy. This can be accomplished with \ampest in which $HT_{A_{\psi},A_{\phi}}$ is used for state preparation and good states are defined as those with $\ket{0}$ in the first register. This requires $O(1/\epsilon)$ queries to $A_{\phi}$ and $A_{\psi}$. 


The Hadamard test can also be used to estimate $Im\big(|\bra{\psi}\ket{\phi}|\big)$ with a slight bit of modification to the original algorithm.
One can see that on slipping an $S^\dagger$-gate after the first $H$-gate in the above algorithm, the final state of the algorithm will be such that the probability of measuring the register $R_1$ to be in $\ket{0}$ is
$$Pr\big[\ket{0}_{R_1}\big] = \Big\| \frac{1}{2} \big(\ket{\psi} -\iota \ket{\phi}\big) \Big\|^2 = \frac{1}{2}\big(1+Im\big(\big| \bra{\psi}\ket{\phi} \big|\big)\big).$$
As before, by estimating $\frac{1}{2}\big(1+Im\big(\big| \bra{\psi}\ket{\phi} \big|\big)\big)$ with $\epsilon/2$ accuracy, it is possible to obtain an estimate of $Im\big(|\bra{\psi}\ket{\phi}|\big)$ with $\epsilon$ accuracy.

\begin{algo}[\trueampest]
\label{algo:trueampest}
Suppose that we are given a circuit $A$ that generates the state $A\ket{0^n} = \ket{\psi}$ and another circuit $A_y$ that generates the state $A_y \ket{0^n} = \ket{y}$. Then, the probability of observing $\ket{0}$ in the first register of the output of the modified Hadamard test can be used to compute $|\Re(\braket{y}{A|0^n})|$. This probability can be estimated with $\epsilon$-additive error using classical means (requiring $O(1/\epsilon^2)$ calls to $A$ and $A_y$), or quantum means such as \aealgo, or \mdaealgo that we describe later, both of which require $O(1/\epsilon)$ calls to $A$ and $A_y$. A similar approach using a slightly modified test (described above) can be used to compute $|\Im(\braket{y}{A|0^n})|$.
\end{algo}


Now, it can be proved using elementary algebra that for any complex number $Z = C+iD$ such that $|Z|^2\le 1$, $\epsilon/\sqrt{2}$-estimates (additive) of $|C|$ and $|D|$ are sufficient to obtain an $\epsilon$-estimate (additive) of $|Z|$. So, the actual absolute value $|\braket{y}{A|0^n}|$ can therefore be also obtained with $\epsilon$-additive error after individually obtaining its real and complex part, both requiring $O(1/\epsilon)$ calls to $A$ and $A_y$.

\section{Multidistribution Amplitude Estimation (\mdistampest)}
\label{sec:mdim-amp-est}


The popularly known amplitude estimation problem is concerned with estimating the probability of a ``good/desired state'', say $\ket{\gamma}$, when the output of a  quantum algorithm $A$ on the input state $\ket{0^l}$ is measured, i.e., estimating $|\braket{\gamma}{A|0^l}|^2$. Traditionally, it takes as input a fixed $\ket{00\ldots 0}$ state, and uses two oracles, $A$ for preparing the state to analyse, and a marking oracle $O_\gamma$ to mark the good state: $O_\gamma = - \ketbra{\gamma} + \sum_{x \not= \gamma} \ketbra{x}$.

There can be two ways to formulate parallel versions of this problem, i.e., in a superposition. First is simultaneous estimation of $|\braket{\gamma^y}{A|0^l}|^2$ for a set of basis states $\{ \ket{\gamma^y} \}_y$. That is, we want to perform the following operation:
$$\sum_y \alpha_y \ket{y}\ket{00\ldots 0} \longrightarrow \sum_y \alpha_y \ket{y} \ket{\bar{\beta}_y},$$
where, $\bar{\beta}_y$s are $\epsilon$-additive estimates of $|\braket{\gamma^y}{A|0^l}|^2$. The complexity would be defined as the number of queries to $A$. Even though this appears to involve a different marking oracle in every branch of the superposition, given that we want to mark only basis states, this can effectively be solved using $\eqae$ discussed earlier in section~\ref{sec:background_qae}.
Parallelization of this kind can be found in the works (\cite{hhlalgo, jeffery2022quantum}) that use variable-time amplitude amplification.

A variant of this problem was recently studied by Apeldoorn in the name of multidimensional amplitude estimation~\cite{van2021quantum} where the goal was to obtain {\em all} the estimates in as many registers, not in a superposition.

For the second manner, 
%
suppose that we are given a family of oracular quantum algorithms $\{A^O_y~:~y=1 \ldots N\}$ each making $k$ queries to an oracle $O$, for some known constant $k$. Let $\ket{\gamma}$ be some fixed state as above. We want to simultaneously estimate $|\braket{\gamma}{A^O_y|0^l}|^2$ for all $y \in [N]$. We will henceforth drop the superscript $O$ and simply write $A_y$, but the reader will find it convenient to remember that the number of calls to $O$ will be our final query complexity measure.

In fact, there is no reason to mandate a fixed $\ket{\gamma}$ for all $k$. Suppose, in addition, we also have a family of states $\{ \ket{\gamma^y} \}_y$ ($\ket{\gamma^y}$ will be called as the $y^{th}$ desired state). In our {\tt Multidistribution Amplitude Estimation} problem (MDistAmpEst), we want to estimate the values of $|\braket{\gamma^y}{A^O_y|0^l}|^2$ for all $y$, again in a superposition as above.
MDistAmpEst combines both the parallelization attempts mentioned earlier.


The MDistAmpEst problem is formally defined below. $n$ will denote $\log(N)$.

\begin{restatable}[Multidistribution Amplitude Estimation Problem (MDistAmpEst)]{problem}{paeprob}
Given $l$-qubit oracular quantum algorithms $\alg_1^O, \alg_2^O, \cdots \alg_N^O$ where each $\alg_i^O$ makes $k$ calls ({\it wlog.}) to an oracle $O$,
and given a family of reflection operators $\{S_{\gamma^y} = \mathcal{I}-2\ket{\gamma^y}\bra{\gamma^y} : 1\le y \le N\}$ where $\gamma^y \in \{0,1\}^l$, we need a circuit for the following operation:
$$\ket{y}\ket{0^l}\ket{0^{\log{1/\epsilon}}} \xrightarrow{}\ket{y}\ket{\phi_y}\ket{\tilde{\beta}_{y\gamy}},$$
where, for each $y$, $\sin^2\Big(\frac{\pi}{2^m}\cdot \tilde{\beta}_{y\gamy}\Big)$ is an estimate of $|\bra{\gamy}A_y^O\ket{0^l}|^2$ such that
$$\bigg|\sin^2\Big(\frac{\pi}{2^m}\cdot \tilde{\beta}_{y\gamy}\Big) - |\bra{\gamy}A_y^O\ket{0^l}|^2\bigg| \le \epsilon/8$$
for some given $0<\epsilon\le 1$ and $m=\log1/\epsilon$. The query complexity of the problem is defined as the total number of calls made to $O$.
\end{restatable}

Suppose that the first register contains the basis state $\ket{y}$. The standard \aealgo algorithm would operate in two stages, first apply $A_y \otimes \iden$, and then the core estimation steps, which we denote $AmpEst_y$; $AmpEst_y$ in-turn uses $A_y$ for preparing the state to process and $S_{\gamma^y}$ to mark the good state. Thus, the overall estimation algorithm looks like the following.
$$AmpEst_y \left[ \left(A_y \ket{0^l}\right) \otimes  \ket{0^{\log 1/\epsilon}} \right] = \ket{\phi_y}\ket{\tilde{\beta}_{y\gamy}},\quad \mbox{ where } \ket{\phi_y} = A_y \ket{0^l}.$$

The solution appears simple --- implement the corresponding conditional operators
$$\mathbf{U}=\sum_y \ketbra{y} \otimes AmpEst_y\mbox{ (conditional-$AmpEst_y$, and)} \quad \mathbf{V}=\sum_y \ketbra{y} \otimes A_y\mbox{ (conditional-$A_y$)}.$$

It is easy to verify that $$\ket{y}\ket{\phi_y}\ket{\tilde{\beta}_{y\gamy}} = \Big(\sum_y \ketbra{y}\otimes AmpEst_y\Big)\Big(\sum_y \ketbra{y}{y} \otimes A_y \otimes \iden^{\log 1/\epsilon}\Big)\cdot \ket{y}\ket{0^l}\ket{0^{\log 1/\epsilon}}.$$


Operators like $\mathbf{U}$ are common in quantum circuits, e.g., they appear in the amplitude estimation circuit as $\sum_y \ketbra{y} \otimes G_f^{2^y}$ where $G_f$ is the Grover's iterator that internally calls the oracle $U_f$.
A naive approach to implement $\mathbf{U}$ is to serially apply the sequence of $\ketbra{y}{y} \otimes AmpEst_y$ operators for each individual $y$. 
Each of these operators would perform amplitude estimation with the good state as $\ket{\gamy}$ and the state preparation oracle as $A_y$, conditioned on the first register being in $\ket{y}$. 
Then, the total number of queries made by $\mathbf{U}$ to the oracle $O$ would be $O(N\frac{k}{\epsilon})$ where $O(k/\epsilon)$ is the query complexity due to a single amplitude estimation $AmpEst^O_y$. 
Similarly, the number of calls made to $O$ for implementing $\mathbf{V}$ is $N$. Thus, the total number of queries to the oracle $O$ would be $O(\frac{Nk}{\epsilon})$.


Our solution here is an algorithm \mdaealgo that performs the same task but with just $O(\frac{k}{\epsilon})$ queries to $O$. The trick is to push the control $\ketbra{y}{y}$ deep inside $AmpEst_y$ and $A_y$.

\begin{restatable}[Multidistribution Amplitude Estimation]{theorem}{simulae_thm}
\label{theorem:simulaetheorem}
    \mdaealgo uses $O(\frac{k}{\epsilon})$ queries to the oracle $O$ to solve the \textsf{Multidistribution Amplitude Estimation Problem}, i.e., with probability at least $\tfrac{8}{\pi^2}$ it outputs
    $\ket{\Phi} = \sum_y \alpha_y\ket{y}\ket{\xi_y}\ket{\tilde{\beta}_{y \gamy}}$ given any initial state $\sum_y \alpha_y \ket{y} \ket{0^l} \ket{0^m}$, where $\sin^2\big(\tilde{\beta}_{y\gamma^y}\frac{\pi}{2^m}\big)$ is an $O(\epsilon)$-estimate of $|\bra{\gamy}A_y^O\ket{0^l}|^2$ for each $y$.
\end{restatable}

The proof of this theorem is present in Appendix~\ref{appendix:algo-mdistampest}.

\section{Amplitude Amplification using Biased Oracle}
\label{sec:aa-biased-oracle}

Given an oracle $\O$, that marks a state of interest (say $\ket{x}$), we know that the amplitude amplification algorithm allows us to obtain $\ket{x}$ {\it w.h.p.} from any given state $\ket{\psi} = A\ket{0}$ quadratically faster as compared to classical approaches using $A$ as a black-box. We use
$AA_{A,\O}$ to denote such an amplitude amplification algorithm, the key ingredient of which is the  Grover iterator $G_{A,\O} = -AR_{\ket{0}}A^{\dagger}\O$; here, $R_{\ket{i}}$ denotes the reflection operator $2\ket{i}\bra{i}-\mathbb{I}$.
The standard assumption is that the oracle $\O$ marks the state $\ket{x}$ with probability 1, i.e., $\O\ket{x}\ket{0} = \ket{x}\ket{1}$ if $x$ is `good' or else, $\ket{x}\ket{0}$.

However, if we replace the oracle $\mathcal{O}$ with a bounded-error oracle $\borcl_p$ which marks $\ket{x}$ but with some probability $p \in (1/2, 1)$, then the naive amplitude estimation algorithm does not work as intended since $\O$ would also mark the other states with probability $1-p$, and this error will accumulate faster with each iteration of the algorithm.
Hoyer et al.~\cite{biasedoraclehoyer} first investigated this setting for $p=9/10$ and showed that by following each amplification step with an error reduction step, it is possible to solve the bounded-error search problem using $O(\sqrt{N})$ queries to oracle $\borcl_p$ for a search of $1$ good element over $N$ elements.
We observed that the algorithm increases the number of ancill\ae\ qubits in each error reduction step, and since, the number of error reduction steps is $O(\sqrt{N})$, the space required for this algorithm becomes $O(\sqrt{N})$ qubits; thus, we cannot use this algorithm as a subroutine for this work.

Another approach, as hinted by the authors of the above work, would be to replace a single call to the oracle $\borcl_p$ in the Grover iterator by the following sub-circuit for marking: make $O(\log(1/\delta))$-many independent calls to $\borcl_p$, then compute the majority over those copies, and finally use the majority value for marking the state of interest.
Naturally, the ancill\ae\ qubits required for performing the majority in each of the $O(\sqrt{N})$ calls is $O(\log(1/\delta))$.
However, not all the ancill\ae~ qubits can be cleaned up for reuse due to their entanglement with the workspace qubits.
Hence, the space complexity would asymptotically remain $O(\sqrt{N})$.

Therefore, we designed another technique that uses just log-space to solve the bounded-error search problem using $O(\sqrt{N}\log(1/\delta))$ queries to $\borcl_p$~\footnote{Careful readers will observe a logarithmic overhead in the query complexity.}. The idea is to replace the operator $A$ in the Grover iterator with a newly constructed operator $\hata$ that internally uses $\borcl_p$ to enhance $A$. The role of $\hata$ will be to generate a state in which the good and the bad states are explicitly marked using an additional register whose state is $\ket{1}$ or $\ket{0}$, accordingly, and furthermore, the probability of the bad state can be made arbitrarily low.

\begin{lemma}
    \label{theorem:biasedaa_construct}
    Suppose that we are given an algorithm $A$ that generates the initial state $A\ket{0^n} = \sum_x \alpha_x \ket{x}$, a bounded-error oracle $\borcl_p$ as defined above and an error parameter $\delta$; further, let $G$ denote the set of good states, and $B$ the set of bad states. Choose an appropriate $k=O\Big(\frac{2p}{(p-1/2)^2}\log(\frac{1}{\delta})\Big)$, and construct a quantum circuit $\hata$ as described in Algorithm~\ref{algo:biasedaa}. Then, (ignoring the states of ancill\ae)
    $$\hata\ket{0^{n+k+1}}=\sum_{x \in G} \alpha_x \ket{x} \Big[ \eta^g_{x0} \ket{\dots} \ket{0} + \eta^g_{x1} \ket{\dots} \ket{1}\Big] + \sum_{x \in B} \alpha_x \ket{x} \Big[\eta^b_{x0} \ket{\dots} \ket{0} + \eta^b_{x1} \ket{\dots} \ket{1} \Big],$$
    such that $|\eta^g_{x0}|^2 \le \delta$ and $|\eta^b_{x1}|^2 \le \delta$.
\end{lemma}


\begin{algorithm}[t]
    \small
	\caption{\label{algo:biasedaa}Constructing the algorithm $\hata_{A,\borcl_p,k}$}
	\begin{algorithmic}[1]
	    \Require Bounded-error oracle $\borcl_p$, the initial algorithm $A$ and $k$.
	    \State Initialize $R_1$ to $\ket{0^n}$. Next $k+1$ registers $R_{21}R_{22}\cdots R_{2k}R_{maj}$ are initialized to $\ket{0}$.
	    \State Apply $A$ to $R_1$.
	    \For{$i$ in $1$ to $k$}
	        \State Apply $\borcl_p$ to $R_1R_{2i}$.
	    \EndFor
	    \State Apply a conditional majority gate, using $R_1$ as the control, using the registers $R_{21} R_{22} \ldots R_{2k}$ as inputs to the majority circuit, and storing the majority value in $R_{maj}$.
	\end{algorithmic}
    \end{algorithm}

The result can be understood by taking $\delta \to 0$, and analysing the observation upon measuring the output of $\hata\ket{0^{n+k+1}}$. For $x \in G$, we are more likely to observe $\ket{x}\ket{\dots}\ket{1}$ as compared to $\ket{x}\ket{\dots}\ket{0}$, and for $x\in B$, $\ket{x}\ket{\dots}\ket{0}$ is the more likely outcome --- i.e., the information about $x$ being good or bad is encoded in the final qubit, {\it w.h.p.}.

The next step is pretty obvious. We run one of the amplification routines using $\hata$ as the state preparation oracle, and amplifying the probability of states like $\ket{x}\ket{\dots}\ket{1}$. This would amplify such states for $x\in G$ but, unfortunately, also for $x \in B$. However, if we choose $\delta$ sufficiently small, we can ensure that the probability of states like $\ket{x}\ket{\dots}\ket{1}$, where $x\in B$, would be extremely small, and hence, would be within tolerable limits even after Algorithm~\ref{algo:errored-amplify}.
    
Let $\lambda$ be the probability obtaining some good state on measuring $\ket{\psi}$ if some good state is present in $\ket{\psi}$; formally, $\lambda = \min_x(|\alpha_x|^2)$ over all ``good'' $x$. We will use the fact that FPAA~\cite{yoder2014fpaa} employed by the algorithm can amplify an unknown success probability, lower bounded by $\lambda$, to any desired $1-\delta$ within $O(\tfrac{1}{\sqrt{\lambda}} \log \tfrac{1}{\delta})$ iterations of Line~\ref{line:FPAA_inside_algo}.

\begin{restatable}[]{theorem}{biasedAAthm}
\label{thm:erroed-oracle-amplification-our}
    Given an $n$-qubit algorithm $A$ that generates the initial state $A\ket{0} = \ket{\psi}=\sum_{x\in \{0,1\}^n}\alpha_x\ket{x}$, a bounded-error oracle $\borcl_p$ as defined above and an error parameter $\delta$, there exists an algorithm that uses $O\Big(\frac{p}{(p-\frac{1}{2})^2\sqrt{\lambda}}\log(\frac{1}{\lambda\delta})\Big)$ queries to $\borcl_p$ along with $n+O\Big(\frac{2p}{(p-1/2)^2}\log(1/\lambda\delta)\Big)$ qubits and outputs a good state with probability at least $1-\delta$, if one exists. If there is no good state in $\ket{\psi}$ then the algorithm outputs ``No Solution'' with probability at least $1-\delta$.
\end{restatable}


\noindent When $p$ is a constant, the number of queries is $\Ot(\tfrac{1}{\sqrt{\lambda}})$ and $\Ot(1)$ additional qubits are used.

\begin{algorithm}[!h]
    \small
	\caption{\label{algo:errored-amplify}Amplitude amplification using a biased oracle}
	\begin{algorithmic}[1]
	    \Require Bounded-error oracle $\borcl_p$, the initial algorithm $A$ and $k$.
	    \State Initialize $k+2$ registers such that the first register $R_1$ is initialize to $\ket{0^n}$ and the next $k+1$ registers $R_{21}R_{22}\cdots R_{2k}R_{maj}$ are initialized to $\ket{0}$.
	    \State Use Algorithm~\ref{algo:biasedaa} to construct $\hata_{A,\borcl_p,k}$. Then, apply $\hata_{A,\borcl_p,k}$ on $R_1R_{21}R_{22}\cdots R_{2k}R_{maj}$.\label{line:FPAA_inside_algo}
	    \State Apply the fixed point amplitude amplification algorithm (FPAA) on $R_{maj}$ using the good state as $\ket{1}$ and with error at most $\delta/2$. Stop if the number of iterations crosses the limit of $\Ot(\tfrac{1}{\sqrt{\lambda}})$ set by the FPAA algorithm.
	    \State Measure $R_{maj}$ as $m$. If $m=\ket{0}$, output ``No Solution''. Else, measure $R_1$ as $y$ and output $y$.
	\end{algorithmic}
    \end{algorithm}

\begin{proof}[Proof Sketch of Theorem~\ref{thm:erroed-oracle-amplification-our}]
    Set $k = O(\frac{2p}{(p-1/2)^2}\log(1/\delta'))$ for a $\delta'=\lambda^4\delta^2$ and construct $\hata_{A,\borcl_p,k}$.
    From Lemma~\ref{theorem:biasedaa_construct}, see that this $\hata$ (dropping the subscripts) behaves as \[\hata\ket{0^n}\ket{0^k}\ket{0} = \sum_x \alpha_x \ket{x} \big(\eta_{x,0}\ket{\phi_{x,0}}\ket{0}+\eta_{x,1}\ket{\phi_{x,1}}\ket{1}\big) = \ket{\Psi}\] where $|\eta_{x,f(x)}|^2 \ge 1-\delta'$ for any $x$; here $f(x)$ indicates the ``goodness'' of $x$.
    
    Now, two cases can happen. We refer the reader to Appendix~\ref{appendix:biased-aa} for the complete proof.
    
    \textbf{Case (i): } Let $f(x)=0$ for all $x$. For this case, we show that in order for the probability of $\ket{1}$ in $R_{maj}$ to become at least $\delta$ post amplification from the initial $\delta'=\lambda^4 \delta^2$, the initial state needs to be amplified $\Omega(\frac{1}{\lambda})$ number of times.
    However, since we perform only $O(\frac{1}{\sqrt{\lambda}})$ many iterations of amplification, the probability of $\ket{1}$ in $R_{maj}$ remains below $\delta$. Thus, the algorithm outputs ``No solution'' with high probability.

    \textbf{Case (ii): } Let $f(x)=1$ for some $x$. In this case, for all $x$ such that $f(x)=1$, we will have $|\eta_{x,1}|^2\ge 1-\delta'$. Then the probability of measuring the last qubit as $\ket{1}$ is at least $\sum_{x : f(x)=1}|\alpha_x \eta_{x,1}|^2 \ge \lambda(1-\delta') > \lambda/2$ (since $\delta < 0.5$).
    Now, using the fixed point amplitude amplification subroutine, in $O(\frac{1}{\sqrt{\lambda}})$ iterations, we obtain a final state post amplification such that with probability $1-\delta$ we obtain $\ket{1}$ on measuring the $R_{maj}$ register.
    
    One more step is needed here. One may argue that even though we obtain $\ket{1}$ with probability $1-\delta$ in $R_{maj}$, the state $\ket{x}$ obtained on measuring $R_1$ need not be such that $f(x)=1$. However, we were able to show that, with probability at least $3/4$ we will obtain a good $x$ upon measuring $R_1$ if we had obtained $\ket{1}$ in $R_{maj}$.
\end{proof}

\section{Probability and Amplitude Filtering}
\label{sec:profil_algo}


The algorithms for both probability filtering and amplitude filtering follow a similar pattern. The first step is to design appropriate biased oracles --- \probfilorcl for probability filtering and \ampfilorcl for amplitude filtering; they are described in \cref{algo:probfilborcl-main} and \cref{algo:ampfilborcl-main}. The oracles are used for marking a basis state to be good if its probability or amplitude, as required, is more than $\tau$. Then, use our amplitude amplification algorithm for biased-oracle (see Section~\ref{sec:aa-biased-oracle}) to amplify the probability of finding a marked state, if one exists. We summarise the behaviour of the filtering algorithms in the following theorems.

\begin{restatable}[Additive-error algorithm for \amprob]{theorem}{ampfilthm}
\label{theorem:ampfil}
The \amprob problem is given as input a $(\log(m)+a)$-qubit quantum algorithm $O_D$ that generates a distribution $D:\big(p_x = |\alpha_x|^2\big)_{x=1}^m$ upon measurement of the first $\log(m)$ qubits of \\\centerline{$\displaystyle O_D \ket{0^{\log(m)+a}} = \sum_{x \in \{0,1\}^{\log(m)}} \alpha_x \ket{x} \ket{\psi_x} = \ket{\Psi}$ (say)}\\ in the standard basis, and also a threshold $\tau \in (0,1)$.

For any choice of parameters $0 < \epsilon < \tau$ for additive accuracy and $\delta$ for error, there exists a quantum algorithm that uses $O\big((\log(m)+\log(\frac{1}{\epsilon})+a)\log(\frac{1}{\delta\tau})\big)$ qubits and makes $O(\frac{1}{\epsilon\tau}\deltaerr)$ queries to $O_D$ such that when its final state is measured in the standard basis, we observe the following. 
\begin{enumerate}
    \item If $|\alpha_x| < \tau - \epsilon$ for all $x$ then the output register is observed in the state $\ket{0}$ with probability at least $1-\delta$.
    \item If $|\alpha_x| \ge \tau$ for any $x$, then with probability at least $1-\delta$ the output register is observed in the state $\ket{1}$ and some $x$ such that $|\alpha_x| \ge \tau$ is returned as output.
\end{enumerate}
\end{restatable}

\begin{restatable}[Additive-error algorithm for \prob]{theorem}{probfilthm}
\label{theorem:probfil}
The \prob filtering problem is given as input a $(\log(m)+a)$-qubit oracle $O_D$ that generates a distribution $D = (p_i)_{i=1}^m$ on the first $\log(m)$ qubits of the state $O_D\ket{0^{\log(m)+a}}$, and also a threshold $\tau$.

For any choice of parameters $0 < \epsilon < \tau$ for additive accuracy and $\delta$ for error, there exists a quantum algorithm that uses $O\big((\log(m)+\log(\frac{1}{\epsilon})+a)\log(\frac{1}{\delta\tau})\big)$ qubits and makes $O(\frac{1}{\epsilon\sqrt{\tau}}\deltaerr)$ queries to $O_D$ such that when its final state is measured in the standard basis, we observe the following. 
\begin{enumerate}
    \item If $p_x < \tau - \epsilon$ for all $x$ then the output register is observed in the state $\ket{0}$ with probability at least $1-\delta$.
    \item If $p_x \ge \tau$ for any $x$, then with probability at least $1-\delta$ the output register is observed in the state $\ket{1}$ and some $x$ such that $p_x \ge \tau$ is returned as output.
\end{enumerate}
\end{restatable}

\begin{algorithm}[!hb]
    \caption{Constructing biased-oracle \probfilorcl for probability filtering \label{algo:probfilborcl-main}}
    \begin{algorithmic}[1]
        \Require Oracle $O_D$ (with parameters $m$, $a$), threshold $\tau$, and accuracy $\epsilon$.
        \Require Input register $R_1$ set to some basis state $\ket{x}$ and output register $R_5$ set to $\ket{0}$. 
        \State Set $r=\log(m)+a$, $\tau' = \frac{1}{2}(1+\tau - \frac{\epsilon}{8})$, $q = \lceil \log(\frac{1}{\epsilon}) \rceil +5$ and $l= q+3$.
        \State Set $\tau_1 = \left\lfloor{\frac{2^l}{\pi}\sin^{-1}(\sqrt{\tau'})}\right\rfloor$
        \State Initialize ancill\ae\ registers $R_2R_3R_4$ of lengths $r, l$ and 1, respectively, and set $R_3 = \ket{\tau_1}$.
        \State {\bf Stage 1:} Apply \eqae (sans measurement) with $R_2$ as the input register, $R_4$ as the output register and $O_D$ is used as the state preparation oracle.
        $R_1$ is used in $EQ$ to determine the ``good state''. \eqae is called with error at most $1- \frac{8}{\pi^2}$ and additive accuracy $\frac{1}{2^q}$.
        \State {\bf Stage 2:} Set $R_5$ to 1 if the estimate of the probability, calculated using $R_4$, meets $\tau$.
        \State \qquad Use {\tt HD$_l$} on $R_3$ and $R_4$ individually.
        \State \qquad Use ${\tt CMP}$ on $R_3 = \ket{\tau_1}$ and $R_4$ as input registers and $R_5$ as output register.
        \State \qquad Use {\tt HD$^{\dagger}_l$} on $R_3$ and $R_4$ individually.
    \end{algorithmic}
\end{algorithm}

\begin{algorithm}[!htbp]
    \caption{Constructing biased-oracle \ampfilorcl for real amplitude filtering \label{algo:ampfilborcl-main}}
    \begin{algorithmic}[1]
        \Require Oracle $O_D$ (with parameters $m$, $a$), threshold $\tau$, and accuracy $\epsilon$.
        \Require Input register $R_1$ set to some basis state $\ket{x}$ and output register $R_5$ set to $\ket{0}$. 
        \State Set $r=\log(m)+a$, $\tau' = \frac{1}{2}(1+\tau - \frac{\epsilon}{8})$, $q = \lceil \log(\frac{1}{\epsilon}) \rceil +5$ and $l= q+3$.
        \State Set $\tau_1 = \left\lfloor{\frac{2^l}{\pi}\sin^{-1}(\sqrt{\tau'})}\right\rfloor$
        \State Initialize ancill\ae\ registers $R_{21}R_{22}R_3R_4$ of lengths $1, r, l$ and $l$, respectively. Set $R_3 = \ket{\tau_1}$.
        \State {\bf Stage 1:} Apply \mdaealgo (sans measurement) with $R_{21}$ as the input register, $R_4$ as the output register and $\ket{0}$ as the ``good state''. A controlled-Hadamard test, i.e., $\sum_y \ketbra{y} \otimes HT_{A_y,O_D}$ over the registers $R_1, R_{21}$ and $R_{22}$ is used as the state preparation oracle, of which only $R_{21}$ is used in \mdaealgo.  \mdaealgo is called with
        error at most $1- \frac{8}{\pi^2}$ and additive accuracy $\frac{1}{2^q}$. Here, $\{A_y ~:~ y \in \{0,1\}^r \}$ is a family of oracles each of which acts as $A_y\ket{0^r}=\ket{y}$.
        
    
        \State {\bf Stage 2:} Set $R_5$ to 1 if the estimate of the probability, calculated using $R_4$, meets $\tau$.
        \State \qquad Use {\tt HD$_l$} on $R_3$ and $R_4$
        individually.
        \State \qquad  Use ${\tt CMP}$ on $R_3 = \ket{\tau_1}$ and $R_4$ as input registers 
        and $R_5$ as output register.
        \State \qquad Use {\tt HD$^{\dagger}_l$} on $R_3$ and $R_4$
        individually.
    \end{algorithmic}
\end{algorithm}

Now we explain how we implemented the biased oracles.
The role of these oracles is to mark the basis states whose probability or amplitude is at least $\tau$, with at most some small error probability. 
Their construction involves two stages: an estimation stage followed by a marking stage.
For the \probfilorcl, we use \eqae to estimate the probabilities of each basis state in superposition.
The construction of the \ampfilorcl differs from that of the \probfilorcl only in the estimation stage. Using \eqae for estimating the amplitudes would lead to a query complexity that depends quadratically on $1/\epsilon$. To circumvent that we replace it with the \mdaealgo which uses \trueampest to estimate the absolute value of the amplitudes of each of the basis states in superposition.
In the marking stage, the algorithm uses straight-forward quantum operations to compare the estimate, in one register, with $\tau$, hardcoded in a suitable encoding in another register. 
Since \eqae and \mdaealgo perform estimation with error probability that is at most $1-\frac{8}{\pi^2}$, both \probfilorcl and \ampfilorcl mark the good basis states with probability at least $8/\pi^2$.

The query complexity arising  from biased-oracle amplitude amplification (see Section~\ref{sec:aa-biased-oracle}) scales as $\Ot(\tfrac{1}{\sqrt{\lambda}})$ where $\lambda = \min |\alpha_x|^2$ among all $x$ that are good. For probability filtering, we want to amplify any state $\ket{x}$ such that $|\alpha_x|^2 \ge \tau$, so, $\lambda \ge \tau$; on the other hand, for amplitude filtering, we want to amplify any state $\ket{x}$ such that $|\alpha_x| \ge \tau$, so $\lambda \ge \tau^2$. Thus, for the former problem, amplitude amplification will call \probfilorcl $\Ot(\tfrac{1}{\sqrt{\tau}})$ times, each of which requires $O(\tfrac{1}{\epsilon})$ calls to $O_D$. Similarly, for the latter problem, \ampfilorcl will be called $\Ot(\tfrac{1}{\tau})$ times, each of which requires $O(\tfrac{1}{\epsilon})$ calls to $O_D$.

Even though the problems are similar, we cannot directly use our algorithm for \prob for solving \amprob. For the latter, we are interested to identify any $x$ such that $|\alpha_x| \ge \tau$. Though this is identical to deciding if $p_x = |\alpha_x|^2$ is greater than the threshold $\tau^2$, to reduce it to \prob we have to set the latter's threshold to $\tau^2$ and estimate $p_x$ with additive accuracy $\epsilon^2$; this leads to the dependency of query complexity on $\epsilon$ as $1/\epsilon^2$. So, to reduce this quadratic dependency to $1/\epsilon$, \ampfilorcl employs the modified Hadamard test described in section~\ref{sec:true-amp-est} and uses \mdaealgo. 
The details of the \probfilorcl and \ampfilorcl algorithms are discussed in Appendix~\ref{appendix:probfilter} and Appendix~\ref{appendix:ampfilter}, respectively.

With access to unbounded space, it is easy to see that one can estimate the distribution $\dist$ with $\epsilon$ additive accuracy using $O(1/\epsilon^2)$ and $O(1/\epsilon)$ queries to the oracle $O_D$ in the classical and quantum settings respectively which can then be used to solve the \prob problem.
However, in the case of $\Tilde{O}(1)$ space, classically, one would be required to make $O(n/\epsilon^2)$ queries to answer the \prob problem.
In contrast, our quantum algorithm solves the same in $\Tilde{O}(1)$ space using just $\Tilde{O}(1/\epsilon\sqrt{\tau})$ queries to $O_D$.
Notice that for constant $\tau$, our algorithm is optimal up to some $\log$ factors.

\section{Lower bounds for \prob and \amprob}
\label{sec:highdist-lb}

For our lower bounds we reduce from the \countdec problem that takes as input a binary string of length $n$ and decides if the number of ones in $X$, denoted $|X|$, is $l_1$ or $l_2 > l_1$, given a promise that one of the two cases is true. We use the following lower bound result proved by Nayak and Wu in~\cite{Nayak1999QuantumStatistics}.

\begin{theorem}[\countdec lower bound~\cite{Nayak1999QuantumStatistics}]\label{thm:nayak_wu}
Let $l_1, l_2 \in [n]$ be two integers such that $l_2 > l_1$, and $X$ be an $n$-bit binary string. Further, let $l_2 - l_1 = 2\Delta$ for some integer $\Delta$. Then any quantum algorithm takes $\Omega(\sqrt{n/\Delta} + \sqrt{(l_2-\Delta)(n-(l_2-\Delta))}/\Delta)$ queries to solve the \countdec problem.
\end{theorem}

\begin{theorem}
\label{thm:prob-lowerbound}
    Any quantum algorithm that solves \prob($\D,\epsilon,\tau$) requires $\Omega(\frac{1}{\epsilon}+\frac{1}{\sqrt{\tau}})$ queries.
\end{theorem}

Formally, let \countdec($l_1, l_2$) be the decision problem of counting if the hamming weight of the given input string is $l_1$ or $l_2$ given the promise that it is one of them.
To prove that \prob requires $\Omega(\tfrac{1}{\epsilon})$ queries, we reduce an instance of \countdec on a $n$-bit string $X$ with $l_1=\tfrac{n}{2}$ and $l_2=\tfrac{n}{2}+\varepsilon n$ to \prob.
Observe that the frequencies of $0$ and $1$ in the string $X$ induce a distribution $\D$. So, an oracle to $X$ can be used to implement an oracle $O_D$ to the distribution $\D$.
By Corollary~1.2 of~\cite{Nayak1999QuantumStatistics}, the query complexity to decide the above \countdec instance is $\Omega(\tfrac{1}{\varepsilon})$. If $|X|=\tfrac{n}{2}$ then $\Pr_D[0]=\Pr_D[1]=\tfrac{1}{2}$, and if $|X|=\tfrac{n}{2}+\varepsilon n$, then $\Pr_D[1]=\tfrac{1}{2}+\varepsilon$. Thus, the output of a \prob algorithm with $\tau=\tfrac{1}{2}+\varepsilon$ and additive accuracy $\epsilon=\varepsilon$ can be used to decide our \countdec instance. This proves a lower bound of $\Omega(\tfrac{1}{\varepsilon})$ for \prob.

Next we prove a bound of $\Omega(\tfrac{1}{\sqrt{\tau}})$.
For proving this lower bound, instead of the \countdec problem, we use the quantum adversary method. We first use this method to obtain a lower bound on the mode decision problem: Given an array $A$ of size $n$ and a threshold $\tau'\in[1,n]$ we have to decide if there exists any element whose frequency is greater than $\tau'$.
We then show a reduction from mode decision problem to \prob to get a lower bound on it.

The main theorem of quantum adversary method can be stated as below~\cite{adversarymethod}:
\begin{theorem}
\label{thm:qa_method}
Let $F$ be a $n$-bit Boolean function and $X$ and $Y$ be two sets of inputs such that $F(x)\neq F(y)$ for any $x\in X$ and $y\in Y$. Let $R\subseteq X\times Y$ be a relation such that 
\begin{enumerate}
    \item for every $x \in X$, $\exists$ at least $m$ different $y \in Y$ such that $(x,y)\in R$.
    \item for every $y \in Y$, $\exists$ at least $ m'$ different $x \in X$ such that $(x,y)\in R$.
    \item for every $x \in X$ and $i \in \{1,...,n\}$, $\exists$ at most $l$ different $y \in Y$ such that $x_i \neq y_i$ and $(x,y)\in R$.
    \item for every $y \in Y$ and $i \in \{1,...,n\}$, $\exists$ at most $l'$ different $x \in X$ such that $x_i \neq y_i$ and $(x,y)\in R$.
\end{enumerate}
Then any quantum algorithm uses $\Omega\Big(\sqrt{\frac{m\cdot m'}{l\cdot l'}}\Big)$ queries to compute $F$ on $X \bigcup Y$.
\end{theorem}

Consider the mode decision problem.  
Let $\tau'\in [1,n]$ be a threshold and set $t = \frac{n}{\tau'-1}$. 
Let $F$ be a Boolean function such that $F(x)=1$ if $x$ is an array whose modal value is greater than or equal to $\tau'$ and $F(y)=0$ if the modal value of $y$ is strictly less than $\tau'$.

Let $Y$ be the set containing one array $B$ such that $B$ contains all unique elements with frequency $\tau'-1$. Let the unique elements be denoted $b_1, b_2, \cdots, b_{t}$.
Let $X$ be the set that contains the arrays $A_{i}$ for all $2\le i \le t$ where $A_i$ is the array that is exactly the same as $B$ except that the first occurrence of $b_i$ is changed to $a_1$.
Notice that the modal element in any $A_i$ is $b_1$.
Define relation $R$ as $R = X\times Y$.

For any $A\in X$, we can see that there is exactly one element $B\in Y$ such that $(A,B)\in R$ since $|Y|=1$. For $B\in Y$, there is exactly $t-1$ elements $A\in X$ such that $(A,B)\in R$ as $|X|=t-1$.
Similarly, for any $A\in X$ and any $i\in [n]$, there is at most one element $B\in Y$ such that $A[i]\neq B[i]$ and $(A,B)\in R$. For $B\in Y$ and any $j\in [n]$, there is at most one element $A\in X$ such that $A[j]\neq B[j]$ and $(A,B)\in R$.

From these we can derive that the quantum query complexity of computing $F$ is $\Omega(\sqrt{\frac{t-1\cdot 1}{1\cdot 1}}) = \Omega(\sqrt{t}) = \Omega(\sqrt{n/\tau'-1})$.

Now, the reduction from the mode decision problem to \prob can be trivially done by setting the threshold of \prob $\tau$ as $\tau = \tau'/n$.
This would imply that the quantum query lower bound of \prob is $\Omega(1/\sqrt{\tau})$ for $\epsilon=\frac{1}{n}$.

\paragraph{} For obtaining the lower bounds of \amprob, it suffices to show that a \prob($\D,2\epsilon,\tau$) instance can be reduced to an \amprob($\D,\epsilon,\sqrt{\tau}$) instance. To show the reduction we prove that the following holds for any $x$:
\begin{itemize}
    \item If $p_x\ge \tau$, then $|\alpha_x|\ge \sqrt{\tau}$.
    \item If $p_x< \tau-2\epsilon$, then $|\alpha_x|< \sqrt{\tau}-\epsilon$.
\end{itemize}

Consider the case when $p_x\ge \tau$. This gives $|\alpha_x|^2 \ge \tau$ which implies $|\alpha_x|\ge \sqrt{\tau}$ proving the first part of the reduction.
Now, let $p_x< \tau-2\epsilon$. This gives that $|\alpha_x|< \sqrt{\tau-2\epsilon}$.
Now, see that
\begin{align*}
    (\sqrt{\tau}-\epsilon)^2 &= \tau + \epsilon^2 - 2\epsilon\sqrt{\tau}\\
    &\ge \tau-2\epsilon\sqrt{\tau}\\
    &\ge \tau-2\epsilon\\
    \implies \sqrt{\tau}-\epsilon &\ge \sqrt{\tau-2\epsilon}
\end{align*}
Using this we have, $|\alpha_x| < \sqrt{\tau-2\epsilon} \le \sqrt{\tau}-\epsilon$ which proves the second part of the reduction.

Using this reduction and Theorem~\ref{thm:prob-lowerbound}, we obtain the following theorem.
\begin{theorem}
\label{thm:amprob-lowerbound}
    Any quantum algorithm that solves \amprob($\D,\epsilon,\tau$) requires $\Omega(\frac{1}{\epsilon}+\frac{1}{\tau})$ queries.
\end{theorem}
\section{Applications of \prob and \amprob}\label{sec:kdist}

\subsection{The \kdist problem}

The \textsc{ElementDistinctness} problem~\cite{Buhrman2005QuantumDistinctness,Ambainis2007QuantumDistinctness,Aaronson2004QuantumProblems} is being studied for a long time both in the classical and the quantum domain. It is a special case of the \kdist problem~\cite{Ambainis2007QuantumDistinctness,Belovs2012Learning-Graph-BasedK-Distinctness} with $k=2$ which too has received a fair attention.

\begin{problem}[\kdist]
Given an oracle to an $n$-sized $m$-valued array $A$, decide if $A$ has $k$ distinct indices with identical values.
\end{problem}

By an $m$-valued array we mean an array whose entries are from $\{0, \ldots, m-1\}$. 
Observe that, \kdist can be reduced to \prob with $\tau=\tfrac{k}{n}$, assuming the ability to uniformly sample from $A$.

The best known classical algorithm for \kdist uses sorting and has a time complexity of $O(n\log(n))$ with a space complexity $O(n)$.
%
In the quantum domain, apart from $k=2$, the $k=3$ setting has also been studied earlier~\cite{Belovs2014ApplicationsAlgorithms,Childs2013AUpdates}.
The focus of all these algorithms has been primarily to reduce their query complexities. As a result their space requirement is significant (polynomial in the size of the list), and beyond the scope of the currently available quantum backends with a small number of qubits.
Recently Li et al.~\cite{Li2019QuantumEstimation} reduced the problem of estimating the min-entropy to \kdist with a very large $k$ making this case additionally important.

The \kdist problem was further generalized to $\Delta$-\gkd  by Montanaro~\cite{Montanaro2016TheMoments} which comes with a promise that either some value appears at least $k$ times or every value appears at most $k-\Delta$ times for a given gap $\Delta$.  The \Finf problem~\cite{Montanaro2016TheMoments,Bun2018ThePolynomials} wants to determine, or approximate, the number of times the most frequent element appears in an array, also known as the modal frequency. Montanaro related this problem to the \gkd problem but did not provide any specific algorithm and left open its query complexity~\cite{Montanaro2016TheMoments}.
So it appears that an efficient algorithm for $\Delta$-\gkd can positively affect the query complexities of all the above problems. However, $\Delta$-\gkd has not been studied elsewhere to the best of our knowledge.

\addtolength{\tabcolsep}{-5pt}

\begin{table}[!ht]
\small
\centering
\caption{Results for the \kdist problem\label{table:summary-kdist}}
\begin{tabular}{ |C{2.5cm}||C{5.8cm}|C{6cm}|  }
 \hline
 \multicolumn{3}{|c|}{\kdist} \\
 \hline
 & Prior upper bound~\cite{Ambainis2007QuantumDistinctness}  & Our upper bound\\
 \hline
 $k\in\{2,3,4\}$ 
    & Setting $r = k$, $O((\frac{n}{k})^{k/2})$ queries, $O(\log(m)+\log(n))$ space &  $\tilde{O}(n^{3/2}/\sqrt{k})$ queries, $O\big((\log(m)+\log(n))\log(\frac{n}{\delta k})\big)$ space\\
\hline
$k=\omega(1)$ and $k\ge 4$ & $O(\frac{n^2}{k})$ queries,\newline $O(\log(m)+\log(n))$ space for $r\ge k$ &  $\tilde{O}(n^{3/2}/\sqrt{k})$ queries, $O\big((\log(m)+\log(n))\log(\frac{n}{\delta k})\big)$ space\\
 \hline
$k=\Omega(n)$ & $O(n^{n/2})$ queries,\newline $O(n\log(m) + \log(n))$ space &  $\tilde{O}(n)$ queries, $O\big((\log(m)+\log(n))\log(\frac{n}{\delta k})\big)$ space\\
\hline
\end{tabular}
\end{table}

\addtolength{\tabcolsep}{5pt}

\subsubsection*{Upper bounds for the \kdist problem}
The $k=2$ version is the \ED problem which was first solved by Buhrman et al.~\cite{Buhrman2005QuantumDistinctness}; their algorithm makes $O(n^{3/4}\log(n))$ queries (with roughly the same time complexity), but requires the entire array to be stored using qubits. A better algorithm was later proposed by Ambainis~\cite{Ambainis2007QuantumDistinctness} using a quantum walk on a Johnson graph whose nodes represent $r$-sized subsets of $[n]$, for some suitable parameter $r \ge k$. He used the same technique to design an algorithm for \kdist as well that uses $\tilde{O}(r)$ qubits and $O(r+(n/r)^{k/2}\sqrt{r})$ queries (with roughly the same time complexity). Later Belovs designed a learning-graph for the \kdist problem, but only for constant $k$, and obtained a tighter bound of $O(n^{\frac{3}{4}-\frac{1}{2^{k+2}-4}})$. It is not clear whether the bound holds for non-constant $k$, and it is often tricky to construct efficiently implementable algorithms base on the dual-adversary solutions obtained from the learning graphs.

Thus it appears that even though efficient algorithms may exist for small values of $k$, the situation is not very pleasant for large $k$, especially $k=\Omega(n)$ --- the learning graph idea may not work (even if the corresponding algorithm could be implemented in a time-efficient manner) and the quantum walk algorithm uses $\Omega(k)$ space. 

We proposed to use \prob to solve \kdist by (a) implementing an oracle $O_D$ from the array $A$ (this is straight forward), and then calling our algorithm for probability filtering using $\tau = k/n$ (see Theorem~\ref{theorem:probfil}, and $\epsilon=1/n$ to ensure that estimates (which are always of the form $t/n$) are well-separated. 

\begin{lemma}\label{lemma:kd-ub}
There exists a bounded-error algorithm for \kdist, for any $k \in [n]$, that uses  $O(\tfrac{n^{3/2}}{\sqrt{k}}\log(\tfrac{1}{\delta \cdot k}))$ queries and $O\big((\log(m) + \log(n))\log(\frac{1}{\delta \cdot k})\big)$ qubits.
\end{lemma}

See Table~\ref{table:summary-kdist} for a comparison of our method with respect to the others.
This algorithm has a few attractive features. It is specifically designed to use $\tilde{O}(1)$ qubits, and as an added benefit, it works for any $k$. Further, it improves upon the algorithm proposed by Ambainis for $k \ge 4$ when we require that $\tilde{O}(1)$ space be used, and moreover, its query complexity does not increase with $k$.
One might be puzzled with the fact that the query complexities are larger than $n$ as it is well known that in unbounded space, the query complexity of \kdist is trivially $n$.
However, when the available space is restricted to $O(\log(n))$ space, then the query complexity of \kdist need not be bounded by $n$.

There have been separate attempts to design algorithms for specific values of $k$. For example, for $k=3$ Belovs designed a slightly different algorithm compared to the above~\cite{Belovs2014ApplicationsAlgorithms} and Childs et al.~\cite{Childs2013AUpdates} gave a random walk based algorithm both of which uses $O(n^{5/7})$ queries and $O(n^{5/7})$ space. These algorithm improved upon the $O(n^{3/2})$-query algorithm proposed earlier by Ambainis~\cite{Ambainis2007QuantumDistinctness}. Our algorithm provides an alternative that matches the query complexity of the latter and can come in handy when a small number of qubits are available.

For $k$ that is large, e.g. $\Omega(n)$, the query complexity of Ambainis' algorithm is exponential in $n$ and that of ours is $O(n^{3/2})$. Montanaro used a reduction from the \countdec problem~\cite{Nayak1999QuantumStatistics} to prove a lower bound of $\Omega(n)$ queries for $k=\Omega(n)$ --- of course, assuming unrestricted space~\cite{Montanaro2016TheMoments}. Our algorithm matches this lower bound, but with only $\tilde{O}(1)$ space.

\subsection{The Non-linearity Estimation Problem}

Non-linearity is an important cryptographic measure of a Boolean function.
Non-linearity of a function $f:\{0,1\}^n \xrightarrow{} \{0,1\}$ is defined in terms of the largest absolute-value of its Walsh-Hadamard coefficient~\cite{Bera2021QuantumEstimation} as 
$$\eta(f) = \tfrac{1}{2} - \tfrac{1}{2} \fmax$$ where $\fmax = \max_x |\hat{f}(x)|$ and $\hat{f}(x)$ is the Walsh-Hadamard coefficient of $f$ at the point $x$.  
Boolean functions with low non-linearity can be easily approximated by linear functions.
One of the conditions that any Boolean function needs to satisfy to be used in cryptographic applications is that its non-linearity is high.
Notice from the definition that to estimate the non-linearity of a function, it suffices to obtain an estimate of $\fmax$ of that function.
Recall that the output state of the Deutsch-Jozsa circuit is $\sum_x \hat{f}(x)\ket{x}$, i.e., the probability of observing $\ket{x}$ is $\hat{f}(x)^2$. 
It immediately follows that we can utilize the \prob algorithm in conjunction with a binary search on the interval $(0,1]$ to estimate $\fmax^2$, and hence, non-linearity, with additive inaccuracy.
This approach is presented as Algorithm 1 in~\cite{Bera2021QuantumEstimation}. 
However, this would lead to a query complexity of $\Ot(1/\lambda^2\fmax)$ \footnote{Although the query complexity of this algorithm has been proved to be $\Ot(1/\lambda^3)$ in~\cite{Bera2021QuantumEstimation}, the query complexity can be reduced to $\Ot(1/\lambda^2\fmax)$ with a slightly tighter analysis of their Algorithm 1.} to estimate non-linearity to within $\lambda$ additive accuracy.

Alternately, we can replace \prob with \amprob to estimate $\fmax$ instead of $\fmax^2$ in the algorithm presented in~\cite{Bera2021QuantumEstimation}. 
This reduces the number of queries since to estimate $\fmax$ within $\pm \lambda$, it now suffices to call \amprob with inaccuracy $\lambda$, instead of calling $\prob$ with inaccuracy $\lambda^2$. 
Given that the query complexity of \amprob is $\Ot(1/\lambda)$, this leads to a quadratic improvement in the query complexity in form of $\Ot(\frac{1}{\lambda\fmax})$. 

\begin{lemma}
    \label{lemma:non-lin}
    Given a Boolean function $f:\{0,1\}^n\xrightarrow{}\{0,1\}$ as an oracle, an accuracy parameter $\lambda$ and an error parameter $\delta$, there exists an algorithm that returns an estimate $\tilde{\eta}_f$ such that $|\eta_f - \tilde{\eta}_f|\le \lambda$ with probability at least $1-\delta$ using $O(\frac{1}{\lambda\hat{f}_{max}}\log(\frac{1}{\lambda})\log(\frac{1}{\delta\hat{f}_{max}}))$ queries to the oracle of $f$. 
\end{lemma}

Bera et al., (\cite{Bera2021QuantumEstimation}) also showed a lower bound of $\Omega(1/\sqrt{\lambda})$ for the non-linearity estimation.
This can further improved to $\Omega(1/\lambda)$ via a reduction from the \countdec problem to the non-linearity problem.

\begin{restatable}{lemma}{nonlinlb}
    \label{lemma:non-lin-lb}
    Any quantum algorithm uses $\Omega(1/\lambda)$ queries to estimate the non-linearity of any given Boolean function.
\end{restatable}

The proof of Lemma~\ref{lemma:non-lin-lb} is given in Appendix~\ref{appendix:non-lin-lb}.
This shows that the modified algorithm for non-linearity estimation that uses \amprob is close to optimal.

\newpage

\bibliographystyle{plain}
\bibliography{ref}

\newpage
\appendix
\section{Amplitude amplification, amplitude estimation and majority}\label{appendix:ampest-and-amp}

\newcommand{\MAJ}{\mathtt{MAJ}}

In this section, we present details on the quantum amplitude amplification subroutine and the $\MAJ$ operator which are used as part of our algorithms.

\subsection{Amplitude amplification}

The amplitude amplification algorithm (AA) is a generalization of the novel Grover's algorithm.
Given an $n$-qubit algorithm $A$ that outputs the state $\ket{\phi}=\sum_k\alpha_k\ket{k}$ on $\ket{0^n}$ and a set of basis states $G=\{\ket{a}\}$ of interest, the goal of the amplitude amplification algorithm is to amplify the amplitude $\alpha_a$ corresponding to the basis state $\ket{a}$ for all $\ket{a}\in G$ such that the probability that the final measurement output belongs to $G$ is close to 1.
In the most general setting, one is given access to the set $G$ via an oracle $O_G$ that marks all the states $\ket{a}\in G$ in any given state $\ket{\phi}$; i.e., $O_G$ acts as
$$O_G \sum_k\alpha_k\ket{k}\ket{0} \xrightarrow{}\sum_{a\notin G}\alpha_a\ket{a}\ket{0} + \sum_{a\in G}\alpha_a\ket{a}\ket{1}.$$

Now, for any $G$, any state $\ket{\phi} = \sum_k\alpha_k\ket{k}$ can be written as 
$$\ket{\phi} = \sum_k\alpha_k\ket{k} = \sin(\theta)\ket{\nu} + \cos(\theta)\ket{\overline{\nu}}$$ where $\sin(\theta) = \sqrt{\sum_{a\in G}|\alpha_a|^2}$, $\ket{\nu} = \frac{\sum_{a\in G}\alpha_a\ket{a}}{\sqrt{\sum_{a\in G}|\alpha_a|^2}}$ and $\ket{\overline{\nu}} = \frac{\sum_{a\notin G}\alpha_a\ket{a}}{\sqrt{\sum_{a\notin G}|\alpha_a|^2}}$.
Notice that the states $\ket{\nu}$ and $\ket{\overline{\nu}}$ are normalized and are orthogonal to each other.
The action of the amplitude amplification algorithm can then be given as
$$AA\Big(\sum_k\alpha_k\ket{k}\ket{0}\Big) = AA\big(\sin(\theta)\ket{\nu} + \cos(\theta)\ket{\overline{\nu}}\big)\ket{0} \xrightarrow{} \sqrt{(1-\beta)}\ket{\nu}\ket{1} + \sqrt{\beta}\ket{\overline{\nu}}\ket{0}$$ where $\beta$ satisfies $|\beta| < \delta$ and $\delta$ is the desired error probability.
This implies that on measuring the final state of AA, the measurement outcome $\ket{a}$ belongs to $G$ with probability $|1-\beta|$ which is at least $1-\delta$.

  \subsection{$\MAJ$ operator}
Let $X_1 \ldots X_k$ be Bernoulli random variables with success probability $p > 1/2$. Let $Maj$ denote their majority value (that appears more than $k/2$ times). Using Hoeffding's bound\footnote{$\Pr[\sum X_i - E[\sum X_i] \ge t] \le \exp(-\frac{2t^2}{n})$}, it can be easily proved that $Maj$ has a success probability at least $1-\delta$, for any given $\delta$, if we choose $k \ge \tfrac{2p}{(p-1/2)^2} \ln \tfrac{1}{\delta}$. We require a quantum formulation of the same.

Suppose we have $k$ copies of the quantum state $\ket{\psi} = \ket{\psi_0} \ket{0} + \ket{\psi_1} \ket{1}$ in which we define ``success'' as observing $\ket{0}$ (without loss of generality) and $k$ is chosen as above. Let $p = \| \ket{\psi_0} \|^2$ denote the probability of success. Suppose we measure the final qubit after applying $(\mathbb{I}^k \otimes MAJ)$ in which the $MAJ$ operator acts on the second registers of each copy of $\ket{\psi}$. Then it is easy to show, essentially using the same analysis as above, that 
$$(\mathbb{I}^k \otimes MAJ) \ket{\psi}^{\otimes k} \ket{0} = \ket{\Gamma_0} \ket{0} + \ket{\Gamma_1} \ket{1}$$ in which $\| \ket{\Gamma_0} \|^2 \ge 1-\delta$.

The $\MAJ$ operator can be implemented without additional queries and with $poly(k)$ gates and $\log(k)$ qubits.
\section{Some Useful Subroutines}
\label{appendix:useful-subroutines}

In this section we present a few subroutines that are used in the construction of \probfilorcl and \ampfilorcl oracles.

\begin{description}
    \item[{\tt EQ$_m$}:] Given two computational basis states $\ket{x}$ and $\ket{y}$ each of $k$ qubits, {\tt EQ$_m$} checks if the {$m$-sized} prefix of $x$ and that of $y$ are equal. Mathematically, {\tt EQ$_m$}$\ket{x}\ket{y} = (-1)^{c}\ket{x}\ket{y}$ where $c=1$ if $x_i=y_i$ for all $i\in [m]$, and $c=0$ otherwise.
    \item[{\tt HD$_q$}:] When the target qubit is $\ket{0^q}$, and with a 
    $q-$bit string $y$ in the control register, {\tt HD} computes the absolute 
    difference of $y_{int}$ from $2^{q-1}$ and outputs it as a string where 
    $y_{int}$ is the integer corresponding to the string $y$. It can be represented 
    as ${\tt HD}_q \ket{y}\ket{b} = \ket{b\oplus\tilde{y}}\ket{y}$ where $y,b\in \{0,1\}^q$ 
    and $\tilde{y}$ is the bit string corresponding to the integer $\abs{2^{q-1} - y_{int}}$.
    Even though the operator {\tt HD} requires two registers, the second register 
    will always be in the state $\ket{0^q}$ and shall be reused by uncomputing (using
    $HD^\dagger$) after the {\tt CMP} gate. For all practical purposes, this operator 
    can be treated as the mapping $\ket{y} \mapsto \ket{\tilde{y}}$.
    \item[{\tt CMP}:] The ${\tt CMP}$ operator is defined as 
    ${\tt CMP} \ket{y_1}\ket{y_2}\ket{b} = \ket{y_1}\ket{y_2}\ket{b \oplus (y_2\le y_1)}$ 
    where $y_1, y_2\in \{0,1\}^n$ and $b\in \{0,1\}$. It simply checks if the integer corresponding to the basis state in the first register is at most that in the second register.
    \item[{\tt Cond-MAJ}:] The ${\tt Cond-MAJ}$ operator is defined as $\prod_x\big(\ket{x}\bra{x}\otimes MAJ\big)$ where $\ket{x}\bra{x}\otimes MAJ$ acts on computational basis states as $MAJ\ket{a_1}\cdots \ket{a_k}\ket{b} = \ket{a_1}\cdots \ket{a_k}\ket{b \oplus (\tilde{a} \ge k/2)}$ where $\tilde{a} = \sum_k a_k$ and $a_i,b \in \{0,1\}$.
\end{description}  
\section{Bounded oracle for amplitude filtering}
\label{appendix:ampfilter}

\subsection{Construction of \ampfilorcl to mark high amplitude states}

\begin{algorithm}[!hb]
    \caption{Constructing biased-oracle \ampfilorcl for real amplitude filtering \label{algo:ampfilborcl-1}}
    \begin{algorithmic}[1]
        \Require Oracle $O_D$ (with parameters $m$, $a$), threshold $\tau$, and accuracy $\epsilon$.
        \Require Input register $R_1$ set to some basis state $\ket{x}$ and output register $R_5$ set to $\ket{0}$. 
        \State Set $r=\log(m)+a$, $\tau' = \frac{1}{2}(1+\tau - \frac{\epsilon}{16})$, $q = \lceil \log(\frac{1}{\epsilon}) \rceil +5$ and $l= q+3$.
        \State Set $\tau_1 = \left\lfloor{\frac{2^l}{\pi}\sin^{-1}(\sqrt{\tau'})}\right\rfloor$
        \State Initialize ancill\ae\ registers $R_{21}R_{22}R_3R_4$ of lengths $1, r, l$ and $l$, respectively. Set $R_3 = \ket{\tau_1}$.
        \State {\bf Stage 1:} Apply \mdaealgo (sans measurement) with $R_{21}$ as the input register, $R_4$ as the output register and $\ket{0}$ as the ``good state''. A controlled-Hadamard test, i.e., $\sum_y \ketbra{y} \otimes HT_{A_y,O_D}$ over the registers $R_1, R_{21}$ and $R_{22}$ is used as the state preparation oracle, of which only $R_{21}$ is used in \mdaealgo.  \mdaealgo is called with
        error at most $1- \frac{8}{\pi^2}$ and additive accuracy $\frac{1}{2^q}$. Here, $\{A_y ~:~ y \in \{0,1\}^r \}$ is a family of oracles each of which acts as $A_y\ket{0^r}=\ket{y}$. \label{line:amp_est}
        
    
        \State {\bf Stage 2:} Set $R_5$ to 1 if the estimate of the probability, calculated using $R_4$, meets $\tau$.
        \State \qquad Use {\tt HD$_l$} on $R_3$ and $R_4$
        individually.\label{line:half_dist}
        \State \qquad  Use ${\tt CMP}$ on $R_3 = \ket{\tau_1}$ and $R_4$ as input registers 
        and $R_5$ as output register.\label{line:q_compare}
        \State \qquad Use {\tt HD$^{\dagger}_l$} on $R_3$ and $R_4$
        individually.\label{line:half_dist_inv}
    \end{algorithmic}
\end{algorithm}

The algorithm is described in Algorithm~\ref{algo:ampfilborcl-1}.
The controlled-Hadamard test on $\ket{x} \ket{0} \ket{0^r}$ is same as $\ket{x} \otimes HT_{A_x,O_D} (\ket{0} \ket{0^r})$, essentially performing the test on $R_{21}$ and $R_{22}$ between $A_x\ket{0^r}=\ket{x}$ and $O_D\ket{0^r}$. Had we measured the output state of $R_{21}$, we would have observed $\ket{0}$ with probability 
$\frac{1}{2}\big(1+Re\big(|\braket{x}{O_D|0^r}|\big)\big)$. \mdaealgo estimates this probability in $R_4$. The steps in Stage 2 are to compare this estimate with the threshold, and set or unset the output register $R_5$ accordingly. After minute analysis of errors, one can prove that the algorithm almost always marks a state if the real-value of its amplitude is above $\tau$.

Extending this algorithm to consider the absolute value of an amplitude follows the expected path of estimating both the real and the complex parts using Hadamard tests (see Section~\ref{sec:true-amp-est}), estimating the absolute value from those two, and then using the latter for marking. 

Further, note that queries to $O_D$ are only made by \mdaealgo, and that too as part of the state-preparation oracle --- the controlled Hadamard test. \mdaealgo makes $O(1/\epsilon)$ calls to the test, and each of them involves one call to $O_D$, leading to an overall query complexity of $O(1/\epsilon)$. The entire behaviour is summarised below.

\begin{restatable}[]{lemma}{ampfilboraclelemma}
\label{lemma:amp-fil-boracle}
    \ampfilorcl makes $O(1/\epsilon)$ queries to $O_D$ and behaves as given below.
    \[\ampfilorcl\ket{x}\ket{0^{2l+r+1}}\ket{0} = \ket{x}\Big(\eta_{x,0}\ket{\phi_{x,0}}\ket{0} + \eta_{x,1}\ket{\phi_{x,1}}\ket{1}\big)\]
    where $|\eta_{x,0}|^2\ge \tfrac{8}{\pi^2}$ if $|\alpha_x| < \tau-\epsilon$ and $|\eta_{x,1}|^2\ge \tfrac{8}{\pi^2}$ if $|\alpha_x| \ge \tau$.
\end{restatable}

Intuitively, \ampfilorcl correctly marks a given $x$ with probability at least $\tfrac{8}{\pi^2}$ if $|\alpha_x|\ge \tau$ and erroneously marks $x$ with probability at most $1-\tfrac{8}{\pi^2}$ if $|\alpha_x| < \tau-\epsilon$, which leads to the claim that \ampfilorcl is a biased oracle with error $1-\tfrac{8}{\pi^2}$.


\subsection{Analysis of \ampfilorcl}

We first assume that all amplitudes are real.

\ampfilboraclelemma*

We require a few technical results to proceed with the proof. These can be proved easily using standard identities.

\begin{proposition}
    \label{prop:fhat-tau-scale-change}
    For any $\alpha_x$, a threshold $\tau$ and some $\epsilon$,
    \begin{enumerate}
        \item $|\alpha_x|\ge \tau+2\epsilon \iff \frac{1}{2}\big(1-|\alpha_x|\big)\le \frac{1}{2}\big(1-\tau\big)-\epsilon$.
        \item $|\alpha_x|< \tau \iff \frac{1}{2}\big(1-|\alpha_x|\big)> \frac{1}{2}\big(1-\tau\big)$.
    \end{enumerate}
\end{proposition}

\begin{proposition}[Proposition~4.1,\cite{Bera2021QuantumEstimation}]\label{lemma:sin_inequality}
For any two angles $\theta_1, \theta_2 \in [0,\pi]$, \[\sin{\theta_1}\le \sin{\theta_2} \iff \sin^2{\theta_1}\le \sin^2{\theta_2} \iff \abs{\frac{\pi}{2}-\theta_1} \ge \abs{\frac{\pi}{2}-\theta_2}.\]
\end{proposition}


\begin{proposition}[Proposition~4.3,\cite{Bera2021QuantumEstimation}]
\label{prop:tau'_tau1_rel}
If two constants $\tau'$ and $\tau_1$ are related as $\tau_1 = \left\lfloor{\frac{2^l}{\pi}\sin^{-1}(\sqrt{\tau'})}\right\rfloor$, then they satisfy $0\le \tau' - \frac{2\pi}{2^l}\le \sin^2(\frac{\pi \tau_1}{2^l})$.
\end{proposition}

\begin{proof}[Proof of Lemma]
    We analyse Algorithm~\ref{algo:ampfilborcl-1} stage by stage for registers $R_1R_{21}R_{22}R_3R_4R_5$ with the input $\ket{x}\ket{0}\ket{0^r}\ket{0^l}\ket{0^l}\ket{0}$.

    \paragraph*{\bf Stage-1:} Consider the registers $R_1R_{21}R_{22}R_4$. We set $R_3 = \ket{\tau_1}$. 
    In Section~\ref{sec:mdim-amp-est} we identified that the \mdaealgo can be given as product of two operators $U = \sum_y\ketbra{y}\otimes AmpEst_y$ and $V = \sum_y\ketbra{y}\otimes B_y \otimes \iden^{l}$.
    In the construction of the oracle, we have $B_y = HT_{A_y, O_D}$ and $A_y$ is such that $A_y\ket{0^r} = \ket{y}$.
    We analyze the state of the register after each of these operations. For simplicity we assume that $O_D\ket{0^r}=\sum_x\alpha_x\ket{x} = \ket{\phi}$ (ignoring the $\ket{\xi_x}$ states).
    
    The state of the registers $R_1, R_{21}$ and $R_{22}$ after applying $V$ on $\ket{x}\ket{0}\ket{0^r}$ can be given as
    \[\ket{\psi_1'} = \frac{1}{2}\ket{x}\big(\ket{0}(\ket{x}+\ket{\phi}) + \ket{1}(\ket{x}-\ket{\phi})\big).\]
    Given that the state in $R_1$ is $\ket{x}$, the probability of obtaining $\ket{0}$ in $R_{21}$ can then be given as \[Pr(\ket{0}_{R_{21}}) = \frac{1}{4}||\ket{x}+\ket{\phi}||^2 = \frac{1}{4}\Big(2+2Re\big(\bra{x}\ket{\phi}\big)\Big) = \frac{1}{2}(1+|\alpha_x|).\]
    Using this, $\ket{\psi_1'}$ can be given as
    \[\ket{\psi_1'} = \ket{x}\Big(\nu_{x0}\ket{0}\ket{\eta_{x0}} + \nu_{x1}\ket{1}\ket{\eta_{x1}}\Big) = \ket{x}\ket{\nu_x}\text{~(say)}\] for some normalized states $\ket{\eta_{x0}}$ and $\ket{\eta_{x1}}$ where $|\nu_{x0}|^2 = \frac{1}{2}(1+|\alpha_x|)$.

    Notice that the operation $U$ is applied on the registers $R_1R_{21}R_{4}$. This operation is given as $\sum_x\ketbra{x} \otimes AmpEst_x$ where $AmpEst_x$ uses the Grover iterator $G_x = -B_xU_{\overline{0}}B_x^{\dagger}U_0$.
    On applying $U$ on $R_1R_{21}R_{4}$, we obtain a state of the registers $R_1R_{21}R_{22}R_3R_4R_5$ as, 
    \begin{align*}
        \ket{\psi_1} &= \ket{x}\ket{\nu_x}\ket{\tau_1}\Big(\beta_{x,s}\ket{a_x}+\beta_{x,\overline{s}}\ket{E_x}\Big)\ket{0}\\
        &= \beta_{x,s}\ket{x}\ket{\nu_x}\ket{\tau_1}\ket{a_x}\ket{0}+\beta_{x,\overline{s}}\ket{x}\ket{\nu_x}\ket{\tau_1}\ket{E_x}\ket{0}\\
        &=  \beta_{x,s}\ket{\psi_{1,s}} + \beta_{x,\overline{s}}\ket{\psi_{1,\overline{s}}}
    \end{align*}
    where $\ket{a_x}$ is a normalized state of the form $\ket{a_x} = \gamma_{+}\ket{a_{x,+}} + \gamma_{-}\ket{a_{x,-}}$ that on measurement outputs $a \in \{a_{x,+}, a_{x,+}\}$ which is
    an $l$-bit string that behaves as $\displaystyle\left| \sin^2\left(\frac{a\pi}{2^l}\right) - |\nu_{x0}|^2 \right| \le \tfrac{1}{2^q}$, $|\beta_{x,s}|^2 \ge 1-\frac{8}{\pi^2}$ and $|\beta_{x,s}|^2 \le \frac{8}{\pi^2}$.
    We denote the set $\{a_{x,+}, a_{x,-}\}$ by $\pmset_{a_x}$.
    Essentially, for any $x$, \mdaealgo stores the correct estimate of the probability of $x$ in $\ket{\nu_x}$ into $R_4$ with probability at least $1-\frac{8}{\pi^2}$ and errors with probability at most $\frac{8}{\pi^2}$.
    
    \paragraph*{\bf Stage-2:} 
Notice that stage 2 affects only the registers $R_3, R_4$ and $R_5$.
Here we are interested only in the state $\ket{\psi_{1,s}}$ that contains the correct estimate.
For any computational basis state $\ket{u}$ and $\ket{v}$, the transformation of a state of the form $\ket{u}\ket{v}\ket{0}$ due to stage 2 can be given as \begin{equation}\label{eqn:stage3}
    \ket{u}\ket{v}\ket{0} \xrightarrow{}\ket{u}\ket{v}\ket{\mathbb{I}\{u\ge v\}} \mbox{ where ~}\mathbb{I}\{u\ge v\}=1 \mbox{ if $u\ge v$ and $0$ else}.
\end{equation}
The reason for indicating ``$u \ge v$'' as 1 and not the other way around is due to the reversal of the direction of the inequality in Proposition~\ref{lemma:sin_inequality}.  Then, stage 2 transforms the state
\[\ket{\psi_{1,s}} = \ket{x}\ket{\nu_x}\ket{\tau_1}\ket{a_x}\ket{0} = \ket{x}\ket{\nu_x}\ket{\tau_1} \Big[ \gamma_{+}\ket{a_{x,+}}\ket{0} + \gamma_{-}\ket{a_{x,-}}\ket{0} \Big]\]
to the state 
\[\ket{\psi_{2,s}} = \ket{x}\ket{\nu_x}\ket{\tau_1} \Big[ \gamma_{+}\ket{a_{x,+}}\ket{\mathbb{I}\{a_{x+}\le \tau_1\}} + \gamma_{-}\ket{a_{x,-}}\ket{\mathbb{I}\{a_{x-}\le \tau_1\}} \Big] \tag{Eqn.~\ref{eqn:stage3}}\]


We will analyse the states $\ket{\mathbb{I}\{a_{x\pm}\le \tau_1\}}$ by considering two types of index $x \in [m]$.

\paragraph*{Scenario (i):} Let $x$ be such that $|\alpha_x| < \tau-\epsilon$. Now, any computational basis state $a$ in $\ket{a_x}$ will be such that $\tilde{a}_x (say) = \sin^2(\frac{a \pi}{2^l}) \in [|\nu_{x0}|^2 -\tfrac{1}{2^q}, |\nu_{x0}|^2 + \tfrac{1}{2^q}]$. Therefore, $\tilde{a}_{x} \le |\nu_{x0}|^2 + \tfrac{1}{2^q} = \frac{1}{2}(1+|\alpha_x|) + \frac{1}{2^q} < \frac{1}{2}(1+\tau-\epsilon) + \tfrac{1}{2^q}$ and since $q$ was chosen such that $2^q \ge \frac{32}{\epsilon}$, $\tilde{a}_x < \frac{1}{2}(1+\tau - \frac{7\epsilon}{16})$.

Since we have $2^l > 2^{q+1} \ge \frac{64}{\epsilon}$, we get that $\frac{2\pi}{2^l} < \frac{2\pi\epsilon}{64} < \frac{6\epsilon}{32}$.
Using this, we have $\frac{1}{2}(1+\tau-\frac{7\epsilon}{16}) = \tau' - \frac{6\epsilon}{32} < \tau'-\frac{2\pi}{2^l} < \sin^2(\frac{\pi \tau_1}{2^l})$ using Proposition~\ref{prop:tau'_tau1_rel}. Since, $\tilde{a}_x < \frac{1}{2}(1+\tau - \frac{7\epsilon}{16})$, we have $\tilde{a}_x = \sin^2(\frac{a\pi}{2^l}) < \frac{1}{2}(1+\tau - \frac{7\epsilon}{16}) < \sin^2(\frac{\pi \tau_1}{2^l})$.

Now on applying \texttt{HD$_l$} on $R_3$ and $R_4$, 
we obtain $\ket{\hat{\tau}_1}$ and $\ket{\hat{a}}$ respectively in $R_3$ and $R_4$ such that $\hat{a} = |2^{l-1} - a|$ and 
$\hat{\tau}_1 = |2^{l-1} - \tau_1|$.
Using Proposition~\ref{lemma:sin_inequality} on the fact that $\tilde{p}_x  = \sin^2(\frac{a\pi}{2^l}) < \sin^2(\frac{\pi \tau_1}{2^l})$ 
we get $\hat{a}  = |2^{l-1} - a| > \hat{\tau}_1  = |2^{l-1} - \tau_1|$.

Since 
$\hat{a} > \hat{\tau}_1$ corresponding to any $\ket{a} \in \{ \ket{a_{x,-}}, \ket{a_{x,+}} \}$, after using \cmp~on $R_3$ and $R_4$, we get in $R_5$ the state $\ket{\mathbb{I}\{a_{x,-}\le \tau_1\}}=\ket{\mathbb{I}\{a_{x,+}\le \tau_1\}}=\ket{0}$. 
Since the state $\ket{\psi_{1,s}}$ which contains the correct estimate exists with probability at least $1-\frac{8}{\pi^2}$ in $\ket{\psi_1}$, in this case, the output state $\ket{\psi_2}$ after stage-2 is such that with probability at least $1-\frac{8}{\pi^2}$ the state $\ket{x}$ is not marked.
In other words, if $x$ is such that $|\alpha_x| < \tau-\epsilon$ then $\ket{\psi_2}$ contains $\ket{x}\ket{\nu_x}\ket{\tau_1}\ket{a_x}\ket{0}$ with probability at least $1-\frac{8}{\pi^2}$.

\paragraph*{Scenario (ii):} $z$ be such that $|\alpha_z| \ge \tau$. 
 Now, any computational basis state $a$ in $\ket{a_z}$ will be such that $\tilde{a}_z (say) = \sin^2(\frac{a \pi}{2^l}) \in [|\nu_{x0}|^2 -\tfrac{1}{2^q}, |\nu_{x0}|^2 + \tfrac{1}{2^q}]$. Therefore, $\tilde{a}_z \ge |\nu_{x0}|^2 - \tfrac{1}{2^q} = \frac{1}{2}(1+|\alpha_z|) - \frac{1}{2^q} \ge \frac{1}{2}(1+\tau) - \frac{1}{2^q}$ and since $q$ was chosen such that $2^q \ge \frac{32}{\epsilon}$, $\tilde{a}_z = \sin^2(\frac{a \pi}{2^l}) > \frac{1}{2}(1+\tau) - \frac{\epsilon}{32} = \tau'$. 
 
This gives us $\sin(\frac{a\pi}{2^l}) > \sqrt{\tau'}$ since $\frac{a\pi}{2^l} \in [0,\pi]$. Furthermore, since $\tau_1$ is an integer in $[0,2^{l}-1]$ and $\tau_1 = \left\lfloor\frac{2^l}{\pi}\sin^{-1}(\sqrt{\tau'})\right\rfloor \le \frac{2^l}{\pi}\sin^{-1}(\sqrt{\tau'})$, we get $\sqrt{\tau'} \ge \sin(\frac{\tau_1\pi}{2^l})$. Combining both the inequalities above, we get $\sin(\frac{a\pi}{2^l}) > \sin(\frac{\tau_1\pi}{2^l})$.

Now, on applying \texttt{HD$_l$} on $R_3$ and $R_4$, 
we obtain $\ket{\hat{\tau}_1}$ and $\ket{\hat{a}}$ respectively in $R_3$ and $R_4$ such that $\hat{a} = |2^{l-1} - a|$ and 
$\hat{\tau}_1 = |2^{l-1} - \tau_1|$. 
Using Proposition~\ref{lemma:sin_inequality} on the fact that $\sin(\frac{a\pi}{2^l}) > \sin(\frac{\tau_1\pi}{2^l})$, we get $\hat{a} < \hat{\tau}_1$.

As above, since $\hat{a} < \hat{\tau}_1$ corresponding to any $\ket{a} \in \{ \ket{a_{x,-}}, \ket{a_{x,+}} \}$, after using \cmp~on $R_3$ and $R_4$, we get in $R_5$ the state $\ket{\mathbb{I}\{a_{x,-}\le \tau_1\}}=\ket{\mathbb{I}\{a_{x,+}\le \tau_1\}}=\ket{1}$.
Again since the state $\ket{\psi_{1,s}}$ exists with probability at least $1-\frac{8}{\pi^2}$ in $\ket{\psi_1}$, the output state $\ket{\psi_2}$ after stage-2 in this case is such that with probability at least $1-\frac{8}{\pi^2}$ the state $\ket{x}$ is marked, i.e, if $x$ is such that $|\alpha_x|\ge \tau$ then $\ket{\psi_2}$ contains $\ket{x}\ket{\nu_x}\ket{\tau_1}\ket{a_x}\ket{1}$ with probability at least $1-\frac{8}{\pi^2}$.

\paragraph*{Query Complexity : } In stage-1, from Theorem~\ref{theorem:simulaetheorem} we get that the number of queries made by \mdaealgo to $O_D$ is $O(1/\epsilon)$. \cmp~and~\hdq~operators does not require any oracle queries. Hence, the total number of queries made to $O_D$ for the construction of $\ampfilorcl$ is $O(1/\epsilon)$.

\end{proof}

\paragraph*{{\normalfont \ampfilorcl} for complex-valued amplitudes:}

The algorithm given for Lemma~\ref{lemma:amp-fil-boracle} works if the amplitudes corresponding to the basis states in $\ket{\Psi} = O_D\ket{r}$ are real.
In the case of complex amplitudes, we need a few additions in the construction of \ampfilorcl.
Alongside estimating the real part of the amplitudes with $\epsilon$ accuracy, we also estimate the complex part of the amplitudes with $\epsilon$ accuracy for all $\ket{x}$ (in superposition) using Algorithm~\ref{algo:trueampest}.
We now use the estimates of the real and the complex parts to compute the norm of the amplitudes for all $\ket{x}$. From the discussions in Section~\ref{sec:true-amp-est}, it is clear that the computed estimates of the absolute value of amplitudes are also $\epsilon$ accurate estimates.
We then use this estimate of the norm to compare against $\ket{\tau_1}$ in stage-3.


\section{Analysis of the algorithm for \amprob}

\ampfilthm*

Intuitively, Lemma~\ref{lemma:amp-fil-boracle} says that for any given $x$, \ampfilorcl correctly marks $x$ with probability at least $1-\delta$ if $|\alpha_x|\ge \tau$ and erroneously marks $x$ with probability at most $\delta$ if $|\alpha_x| < \tau-\epsilon$.
Notice that \ampfilorcl is a biased oracle with error $\delta$. Our algorithm for amplitude filtering makes use of that.

\begin{proof}
    Initialize registers $R_1R_{21}R_{22}R_3R_4R_5$ as $\ket{0^r}\ket{0^l}\ket{0^r}\ket{0^l}\ket{0}$ and apply $O_D$ on $R_1$. The state of the registers can be given as
    \[\ket{\psi} = \sum_{x\in [m]}\alpha_x\ket{x, \xi_x}\ket{0^l}\ket{0^r}\ket{0^l}\ket{0}\]
    
    Next, apply $\ampfilorcl$ on the registers. Finally, use the amplitude amplification algorithm for biased oracle given as Algorithm~\ref{algo:errored-amplify} with $\ket{1}$ in $R_5$ as the good state. Now, if there is a good $x$ in $\ket{\psi}$, then we know that $|\alpha_x|\ge \tau$ or $|\alpha_x|^2 \ge \tau^2$.  Using Theorem~\ref{thm:erroed-oracle-amplification-our} with $p = 1-\frac{8}{\pi^2}$ and Lemma~\ref{lemma:amp-fil-boracle}, we obtain the required proof.
\end{proof}

\section{Bounded oracle for probability filtering}
\label{appendix:probfilter}

\subsection{Construction of \probfilorcl to mark high probability states}

\begin{algorithm}[!htbp]
    \caption{Constructing biased-oracle \probfilorcl for probability filtering \label{algo:probfilborcl}}
    \begin{algorithmic}[1]
        \Require Oracle $O_D$ (with parameters $m$, $a$), threshold $\tau$, and accuracy $\epsilon$.
        \Require Input register $R_1$ set to some basis state $\ket{x}$ and output register $R_5$ set to $\ket{0}$. 
        \State Set $r=\log(m)+a$, $\tau' = \frac{1}{2}(1+\tau - \frac{\epsilon}{8})$, $q = \lceil \log(\frac{1}{\epsilon}) \rceil +5$ and $l= q+3$.
        \State Set $\tau_1 = \left\lfloor{\frac{2^l}{\pi}\sin^{-1}(\sqrt{\tau'})}\right\rfloor$
        \State Initialize ancill\ae\ registers $R_2R_3R_4$ of lengths $r, l$ and 1, respectively, and set $R_3 = \ket{\tau_1}$.
        \State {\bf Stage 1:} Apply \eqae (sans measurement) with $R_2$ as the input register, $R_4$ as the output register and $O_D$ is used as the state preparation oracle.
        $R_1$ is used in $EQ$ to determine the ``good state''. \eqae is called with error at most $1- \frac{8}{\pi^2}$ and additive accuracy $\frac{1}{2^q}$.\label{line:amp_est_probfil}
        \State {\bf Stage 2:} Set $R_5$ to 1 if the estimate of the probability, calculated using $R_4$, meets $\tau$.
        \State \qquad Use {\tt HD$_l$} on $R_3$ and $R_4$ individually.\label{line:half_dist_probfil}
        \State \qquad Use ${\tt CMP}$ on $R_3 = \ket{\tau_1}$ and $R_4$ as input registers and $R_5$ as output register.\label{line:q_compare_probfil}
        \State \qquad Use {\tt HD$^{\dagger}_l$} on $R_3$ and $R_4$ individually.\label{line:half_dist_inv_probfil}
    \end{algorithmic}
\end{algorithm}

The algorithm is described in Algorithm~\ref{algo:probfilborcl}. It uses the following two subroutines.

\begin{description}
    \item[{\tt HD$_q$}:] When the target qubit is $\ket{0^q}$, and with a 
    $q-$bit string $y$ in the control register, {\tt HD} computes the absolute 
    difference of $y_{int}$ from $2^{q-1}$ and outputs it as a string where 
    $y_{int}$ is the integer corresponding to the string $y$. It can be represented 
    as ${\tt HD}_q \ket{y}\ket{b} = \ket{b\oplus\tilde{y}}\ket{y}$ where $y,b\in \{0,1\}^q$ 
    and $\tilde{y}$ is the bit string corresponding to the integer $\abs{2^{q-1} - y_{int}}$.
    Even though the operator {\tt HD} requires two registers, the second register 
    will always be in the state $\ket{0^q}$ and shall be reused by uncomputing (using
    $HD^\dagger$) after the {\tt CMP} gate. For all practical purposes, this operator 
    can be treated as the mapping $\ket{y} \mapsto \ket{\tilde{y}}$.
    \item[{\tt CMP}:] The ${\tt CMP}$ operator is defined as 
    ${\tt CMP} \ket{y_1}\ket{y_2}\ket{b} = \ket{y_1}\ket{y_2}\ket{b \oplus (y_2\le y_1)}$ 
    where $y_1, y_2\in \{0,1\}^n$ and $b\in \{0,1\}$. It simply checks if the integer corresponding to the basis state in the first register is at most that in the second register.
\end{description}  


In Stage 1 of the algorithm, \eqae estimates the probability of $\ket{x}$ in $O_D\ket{0^r}$ in $R4$, and in Stage 2, this estimate is compared with $\tau$ to set or unset $R_5$. This makes \probfilorcl a biased oracle with error $1-\tfrac{8}{\pi^2}$. Further, observe that \eqae makes $O(1/\epsilon)$ calls to the state preparation oracle, in this case, $O_D$, and no one else adds to this. The overall behaviour is summarised below.

\begin{restatable}[]{lemma}{probfilboraclelemma}
    \label{lemma:prob-fil-boracle}
    \probfilorcl makes $O(1/\epsilon)$ calls to $O_D$.
    Upon measuring its output on $\ket{x}\ket{0^{2l+r+1}}\ket{0}$, we observe the following with probability at least $\tfrac{8}{\pi^2}$.
    \[\probfilorcl\ket{x}\ket{0^{2l+r+1}}\ket{0} \Longrightarrow
    \begin{cases}
        \ket{x}\ket{\phi_x}\ket{0} & \text{, if } p_x < \tau-\epsilon,\\
        \ket{x}\ket{\phi_x}\ket{1} & \text{, if } p_x \ge \tau.\\
    \end{cases} \]
\end{restatable}

\subsection{Analysis of \probfilorcl}

\probfilboraclelemma*

Observe that the algorithm \probfilorcl differs from that of \ampfilorcl only at Stage 1. Stage 2 is identical.

\begin{proof}
    We analyse the algorithm on the input state on registers $R_1R_{21}R_{22}R_3R_4R_5$ as $\ket{x}\ket{0^r}\ket{0^l}\ket{0^l}\ket{0}$ where $t=2l+r+1$.
    We set $R_3 = \ket{\tau_1}$.
    On applying $\eqae_{O_D}$ on $R_1R_2R_4$ with $R_2$ as the input register and $R_4$ as the output register and $R_1$ for marking the ``good'' state whose probability we desire to estimate (using, of course, the $EQ$ oracle), the input state transforms to
    \begin{align*}
        \ket{\psi_1} &= \ket{x}\ket{\Psi} \Big(\beta_{x,s}\ket{a_x}+\beta_{x,\overline{s}}\ket{E_x}\Big)\ket{\tau_1}\ket{0}\\
        &= \beta_{x,s}\ket{x}\ket{\Psi} \ket{a_x}\ket{\tau_1}\ket{0}+\beta_{x,\overline{s}}\ket{x}\ket{\Psi}\ket{E_x}\ket{\tau_1}\ket{0}\\
        &= \beta_{x,s}\ket{\psi_{1,s}} + \beta_{x,\overline{s}}\ket{\psi_{1,\overline{s}}}
    \end{align*}

    where $\ket{a_x}$ is a normalized state of the form $\ket{a_x} = \gamma_{+}\ket{a_{x,+}} + \gamma_{-}\ket{a_{x,-}}$ that on measurement outputs $a \in \{a_{x,+}, a_{x,+}\}$ which is
    an $l$-bit string that behaves as $\displaystyle\left| \sin^2\left(\frac{a\pi}{2^l}\right) - |\alpha_x|^2 \right| \le \tfrac{1}{2^q}$, $|\beta_{x,s}|^2 \ge \frac{8}{\pi^2}$ and $|\beta_{x,s}|^2 \le 1-\frac{8}{\pi^2}$.
    
    We denote the set $\{a_{x,+}, a_{x,-}\}$ by $\pmset_{a_x}$.
    Essentially, for any $x$, \eqae stores the correct estimate of the probability of $x$ in $\ket{\Psi}$ into $R_4$ with probability at least $\frac{8}{\pi^2}$.
    
    The correctness of stage-2 is exactly the same as that in the proof for Theorem~\ref{lemma:amp-fil-boracle}.
    
    \paragraph*{Query Complexity : } All calls to $O_D$ are made by $\eqae$ and the latter's query complexity is $O(1/\epsilon)$.
\end{proof}

\section{Algorithm for {\normalfont \mdistampest} = $\mathbf{U} \cdot \mathbf{V}$}
\label{appendix:algo-mdistampest}

First, we discuss how to implement the conditional $A^O_y$ operator $\mathbf{V} = \sum_y \ketbra{y} \otimes A^O_y$, that  acts on two registers, and operates $A^O_y$ on the second register when the first register is in the basis state $y$ --- without loss of generality, we can assume that $\{y\}$ forms the standard basis.
Let $m$ denote the precision of amplitude estimation as stated above, i.e., $\log 1/\epsilon$.

Observe that for each $y$, $A^O_y$ can be expressed as $A^O_y = U_{(k,y)} O U_{(k-1,y)} \cdots U_{(1,y)} O U_{(0,y)}$ with suitable $U_{(i,y)}$ unitaries.  Controlling a sequence of gates is equivalent to a sequence of controlled-gates. So, we can express the conditional $A^O$ operator in the following manner. 

\begin{align*}
    & \sum_y \ketbra{y}{y} \otimes A^O_y = \sum_y \ketbra{y} \otimes  U_{(k,y)} O U_{(k-1,y)} \cdots U_{(1,y)} O U_{(0,y)} \\
    = & \big( \sum_y \ketbra{y} \otimes U_{(k,y)} \big) \big( \sum_y \ketbra{y} \otimes O \big) \big( \sum_y \ketbra{y} \otimes U_{(k-1,y)} \big) \ldots \big( \sum_y \ketbra{y} \otimes O \big) \big( \sum_y \ketbra{y} \otimes U_{(0,y)} \big)\\
    =~ & \big( \sum_y \ketbra{y} \otimes U_{(k,y)} \big) \big( \iden \otimes O \big) \big( \sum_y \ketbra{y} \otimes U_{(k-1,y)} \big) \ldots \big( \iden \otimes O \big) \big( \sum_y \ketbra{y} \otimes U_{(0,y)} \big) \\
    =~ & U'_{(k,y)} \circ \big( \iden \otimes O) \circ  U'_{(k-1,y)} \circ \big( \iden \otimes O \big) \circ U'_{(k-2,y)}  \ldots U'_{(1,y)} \circ \big( \iden \otimes O \big) \circ U'_{(0,y)}
\end{align*}
in which the operator $U'_{t,y}$ denotes $\big( \sum_y \ketbra{y} \otimes U_{(t,y)} \big)$ that {\em does not} involve any call to $O$. Therefore, the query complexity of one call to $V$ is $k$ --- the query complexity of any $A^O$. We want to note here that in the case where each $A^O_y$ makes different number of calls to $O$, $k$ can be taken to be the maximum of the individual query complexities since each $A^O_y$ can be suitably padded to include more calls to $O$. We also point out that the reduction in query complexity comes at the expense of additional non-$O$ gates (e.g., gates required to implement $U'_{(k,y)}$), the number of which can even be exponential in $n$.

The idea above can be extended to the implementation of $\mathbf{U}=\sum_y \ketbra{y} \otimes AmpEst_y$ as well. The amplitude estimation operator due to Brassard et al.~\cite{brassard2002quantum}, excluding the initial state preparation, can be expressed as $AmpEst = (F_m^{-1} \otimes \iden)\cdot \Lambda_m(G) \cdot (F_m \otimes \iden)$ where $F_m$ is the Fourier transform on $m$ qubits, $\Lambda_m(G)$ is the conditional operator defined as $\sum_x \ketbra{x}\otimes G^x$, $G = -AS_{0}A^{\dagger}S_{\chi}$ is the Grover iteator, $G^x$ implies that the $G$ operator is applied $x$ times in succession, and $S$ are reflection operators which do not involve $O$. In a similar manner, $AmpEst_y$ can be expressed as $AmpEst_y = (F_m^{-1} \otimes \iden)\cdot \Lambda_m(G_y) \cdot (F_m \otimes \iden)$ where $G_y = -A_yS_0A_y^{\dagger}M_{\gamma^y}$.

Let's focus on the controlled version of $AmpEst_y$. Controlling a sequence of gates is equivalent to a sequence of controlled-gates.
That is, 
\begin{align*}
    \mathbf{U} &= \sum_y \ketbra{y}\otimes AmpEst_y\\
    &= \sum_y \ketbra{y}\otimes \Big((F_m^{-1} \otimes \iden)\circ \Lambda_m(G_y) \circ (F_m \otimes \iden)\Big)\\
    &= \Big(\sum_y \ketbra{y}\otimes (F_m^{-1} \otimes \iden)\Big)\circ \Big(\sum_y \ketbra{y}\otimes \Lambda_m(G_y)\Big)\circ \Big(\sum_y \ketbra{y}\otimes (F_m \otimes \iden)\Big)\\
    &= \Big(\sum_y \ketbra{y}\otimes (F_m^{-1} \otimes \iden)\Big)\circ \Big(\sum_y \ketbra{y}\otimes \sum_x \ketbra{x}\otimes G_y^x\Big)\circ \Big(\sum_y \ketbra{y}\otimes (F_m \otimes \iden)\Big).
\end{align*}

Both the first and the third operator from the last expression do not involve any calls to $O$ and can be ignored for the purpose for query complexity. Using similar ideas as above, we show how to implement the middle operator efficiently in the number of queries, i.e, independent of the number of state preparation oracles.

\begin{lemma}\label{lemma:mdistae}
    The operator $W=\sum_y \ketbra{y}\otimes \sum_x \ketbra{x}\otimes G_y^x$ can be implemented using $O(k2^m)$ calls to $O$ where $m$ is the size of the last register, $G_y = -A^O_yS_0A_y^{O\dagger}M_{\gamma^y}$, and $k$ denotes the query complexity of $A^O_y$.
\end{lemma}



To prove Lemma~\ref{lemma:mdistae}, we require the following technical result on conditional operators. Here $\mbc{i,p}(U)$ denotes the operator $\iden^{i-1}\otimes \ketbra{p}\otimes \iden^{m-i} \otimes U$. The proof of the below result is straightforward and algebraic in nature, and is available in the Appendix~\ref{appendix:mdimae-components}.



\begin{restatable}[]{lemma}{cproductlemma}
\label{lemma:cproduct_2}
    For any two unitaries $A$ and $B$, we have    $\sum_{y}\ketbra{y}\otimes \big[\mbc{i,1}(A\circ B)+\mbc{i,0}(\iden)\big] = \Big\{\sum_{y}\ketbra{y}\otimes \big[\mbc{i,1}(A)+\mbc{i,0}(\iden)\big]\Big\}\circ \Big\{\sum_{y}\ketbra{y}\otimes \big[\mbc{i,1}(B)+\mbc{i,0}(\iden)\big]\Big\}$
\end{restatable}


\begin{proof}[Proof of Lemma~\ref{lemma:mdistae}]
We first simplify an important component of $W$~\cite{brassardestimation}.
\begin{align*}
   & \sum_x \ketbra{x}\otimes G_y^x = \prod_{i=1}^m \bigg(\Big[\iden^{i-1}\otimes \ketbra{1}\otimes \iden^{m-i} \otimes G_y^{2^i}\Big] + \Big[\iden^{i-1}\otimes \ketbra{0}\otimes \iden^{m-i} \otimes \iden \Big]\bigg)\\
   = & \prod_{i=1}^m \iden^{i-1} \otimes \left( \ketbra{0} \otimes \iden^{m-i} \otimes \iden + \ketbra{1} \otimes \iden^{m-i} \otimes G_y^{2^i} \right)
\end{align*}
that is essentially a sequence of conditional-$G_y^{2^i}$ gates, conditioned on the $i$-th qubit of the first register of this operator.

This allows us to rewrite $W$ in the following manner.
\begin{align}
    W=&\sum_y \ketbra{y}\otimes \sum_x \ketbra{x}\otimes G_y^x\nonumber\\ 
    =& \prod_{i=1}^m\bigg[\sum_y \ketbra{y}\otimes \bigg(\Big[\iden^{i-1}\otimes \ketbra{1}\otimes \iden^{m-i} \otimes G_y^{2^i}\Big] + \Big[\iden^{i-1}\otimes \ketbra{0}\otimes \iden^{m-i} \otimes \iden \Big]\bigg)\bigg]\nonumber\\
    & \mbox{(using the notation of $\mbc{i,p}(U)$)}\nonumber\\
    =& \prod_{i=1}^m\bigg[  \sum_y \ketbra{y}\otimes \Big(\mbc{i,1}\Big(G_y^{2^i}\Big) + \mbc{i,0}(\iden) \Big)  \bigg] = \prod_{i=1}^m W'_i \quad \mbox{(say)}\label{eqn:mid-operator-one-term}
\end{align}


Now see that for any $i$, $W'_i=$
\begin{align}
    &\sum_y \ketbra{y}\otimes \Big(\mbc{i,1}\Big((-A_yS_0A_y^{\dagger}M_{\gamma^y})^{2^i}\Big) + \mbc{i,0}(\iden) \Big)\tag{now use Lemma~\ref{lemma:cproduct_2}}\nonumber\\
    =& \Bigg[\bigg\{\sum_y \ketbra{y}\otimes \Big(\mbc{i,1}\big(-A_y\big) + \mbc{i,0}(\iden) \Big)\bigg\} \circ 
    \bigg\{\sum_y \ketbra{y} \otimes \Big(\mbc{i,1}\big(S_0\big) + \mbc{i,0}(\iden) \Big)\bigg\} ~\circ\nonumber\\
    &\hspace{1cm} \bigg\{\sum_y \ketbra{y}\otimes \Big(\mbc{i,1}\big(A_y^{\dagger}\big) + \mbc{i,0}(\iden) \Big)\bigg\} \circ 
    \bigg\{\sum_y \ketbra{y} \otimes \Big(\mbc{i,1}\big(M_{\gamma^y}\big) + \mbc{i,0}(\iden) \Big)\bigg\}\Bigg]^{2^i}\nonumber\\
    =& \Bigg[\bigg\{\sum_y \ketbra{y}\otimes \Big(\mbc{i,1}\big(-A_y\big) + \mbc{i,0}(\iden) \Big)\bigg\} \circ 
    \bigg\{\iden^n \otimes \Big(\mbc{i,1}\big(S_0\big) + \mbc{i,0}(\iden) \Big)\bigg\} ~\circ\nonumber\\
    &\hspace{1cm} \bigg\{\sum_y \ketbra{y}\otimes \Big(\mbc{i,1}\big(A_y^{\dagger}\big) + \mbc{i,0}(\iden) \Big)\bigg\}\circ 
    \bigg\{\iden^n \otimes \Big(\mbc{i,1}\big(M_{\gamma^y}\big) + \mbc{i,0}(\iden) \Big)\bigg\}\Bigg]^{2^i}\nonumber\\
    =& \Bigg[\bigg\{\sum_y \ketbra{y}\otimes \Big(\mbc{i,1}\big(-A_y\big) + \mbc{i,0}(\iden) \Big)\bigg\} \circ 
    D_i ~\circ \bigg\{\sum_y \ketbra{y}\otimes \Big(\mbc{i,1}\big(A_y^{\dagger}\big) + \mbc{i,0}(\iden) \Big)\bigg\}\circ E_i \Bigg]^{2^i} \nonumber
\end{align}
in which the $D_i$ and $E_i$ operators can be implemented without any calls to $O$.


Next, since we have $A_y = U_{(k,y)}OU_{(k-1,y)}\cdots U_{(1,y)}OU_{(0,y)}$, we can write
\begin{align}
    & \sum_y \ketbra{y}\otimes \Big(\mbc{i,1}\big(-A_y\big) + \mbc{i,0}(\iden) \Big)\label{eqn:controlled-algo}\\
    =& -\sum_y \ketbra{y}\otimes \Big(\mbc{i,1}\big(U_{(k,y)}OU_{(k-1,y)}\cdots U_{(1,y)}OU_{(0,y)}\big) + \mbc{i,0}(\iden) \Big)\nonumber\\
    =& -\prod_{j=k}^{1}\Bigg[\bigg\{\sum_y \ketbra{y}\otimes \bigg( \mbc{i,1}\big(U_{(j,y)}\big) + \mbc{i,0}(\iden)\bigg)\bigg\}\circ\bigg\{\sum_y \ketbra{y}\otimes \bigg(\mbc{i,1}(O) + \mbc{i,0}(\iden)\bigg)\bigg\}\Bigg] \circ\nonumber\\
    & \hspace{5cm}\sum_y \ketbra{y}\otimes \bigg(\mbc{i,1}\big(U_{(0,y)}\big) + \mbc{i,0}(\iden)\bigg)\nonumber\\
    =& -\prod_{j=k}^{1}\Bigg[\bigg\{\sum_y \ketbra{y}\otimes \bigg( \mbc{i,1}\big(U_{(j,y)}\big) + \mbc{i,0}(\iden)\bigg)\bigg\}\circ\bigg\{\iden^{n}\otimes \bigg(\mbc{i,1}(O) + \mbc{i,0}(\iden)\bigg)\bigg\}\Bigg] \circ\nonumber\\
    & \hspace{5cm}\sum_y \ketbra{y}\otimes \bigg(\mbc{i,1}\big(U_{(0,y)}\big) + \mbc{i,0}(\iden)\bigg)\nonumber\label{eqn:final_sequence}
\end{align}
Each of the $\Big[\sum_y \ketbra{y}\otimes \Big( \mbc{i,1}\big(U_{(j,y)}\big) + \mbc{i,0}(\iden)\Big)\Big]$ terms can be implemented as 
$$\prod_{y=0}^{N-1} \bigg\{\Big[\ketbra{y}\otimes \Big( \mbc{i,1}\big(U_{(j,y)}\big) + \mbc{i,0}(\iden)\Big)\Big] + \sum_{x\neq y}\ketbra{x}\otimes \iden^m \otimes \iden\bigg\}$$ which can be identified as a sequence of $N$ controlled gates that do not involve $O$.

The operator $\Big[\iden^{n}\otimes \Big(\mbc{i,1}(O) + \mbc{i,0}(\iden)\Big)\Big]$ is applied independent of the state in the first register.
So, it can be implemented as a controlled-$O$ operation controlled by the $i$th qubit of the second register and that uses only 1 oracle query.

Therefore, we can can see that the number of oracle queries required to implement $\Big[\sum_y \ketbra{y}\otimes \Big(\mbc{i,1}\big(-A_y\big) + \mbc{i,0}(\iden) \Big)\Big]$ (operator in Equation~\ref{eqn:controlled-algo}) is exactly $k$. The query complexity of the operator $\Big[\sum_y \ketbra{y}\otimes \Big(\mbc{i,1}\big(A_y^{\dagger}\big) + \mbc{i,0}(\iden) \Big)\Big]$ is also $k$ using the same analysis. From this we get the query complexity of $W'_i$ as $2^i\cdot2k$.

Now, using Equation~\ref{eqn:mid-operator-one-term}, the total number of oracle queries required for $W$ is $2^{m+2}k$.
Since we have set $m=O(\log(1/\epsilon))$ we get the query complexity of $\mdaealgo$ as $O(k/\epsilon)$.
\end{proof}
\section{Useful Lemme to prove Lemma~\ref{lemma:cproduct_2}}
\label{appendix:mdimae-components}

\begin{lemma}\label{lemma:product_distributive}
    Let $\{A_y\}$ and $\{B_y\}$ be two sets of indexed unitaries. Then, $$\sum_y\ketbra{y} \otimes (A_y\circ B_y) = \Big(\sum_y\ketbra{y} \otimes (A_y)\Big)\circ \Big(\sum_z\ketbra{z} \otimes (B_z)\Big).$$
\end{lemma}
\begin{proof}
    \begin{align*}
        \Big(\sum_y\ketbra{y} \otimes A_y\Big)\circ \Big(\sum_z\ketbra{z} \otimes B_z\Big) &=\sum_{y,z(}\ketbra{y}\circ\ketbra{z}) \otimes (A_y\circ B_z)\\
        &=\sum_y\ketbra{y}\otimes (A_y\circ B_y)
    \end{align*}
\end{proof}

\begin{lemma}\label{lemma:cproduct_1}
    For any two unitaries $A$ and $B$, we have $$\mbc{i,p}(A)\circ \mbc{i,q}(B) = \delta_{p,q}\cdot \mbc{i,p}(A\circ B)$$ where $\delta_{p,q} = 1$ if $p=q$ and $0$ otherwise.
\end{lemma}
\begin{proof}
    \begin{align*}
        \mbc{i,p}(A)\circ \mbc{i,q}(B) &= \big(\iden^{i-1}\otimes \ketbra{p}\otimes \iden^{m-i} \otimes A\big)\circ \big(\iden^{i-1}\otimes \ketbra{q}\otimes \iden^{m-i} \otimes B\big)\\
        &= \iden^{i-1}\otimes (\ketbra{p}\circ\ketbra{q}) \otimes \iden^{m-i} \otimes (A\circ B)\\
        &= \iden^{i-1}\otimes \delta_{p,q}(\ketbra{p}) \otimes \iden^{m-i} \otimes (A\circ B)\\
        &= \delta_{p,q} \big(\iden^{i-1}\otimes \ketbra{p} \otimes \iden^{m-i} \otimes (A\circ B)\big)\\
        &= \delta_{p,q}\cdot\mbc{i,p}(A\circ B)
    \end{align*}
\end{proof}

\cproductlemma*

\begin{proof}
    \begin{align*}
        &\sum_{y}\ketbra{y}\otimes \big[\mbc{i,1}(A\circ B)+\mbc{i,0}(\iden)\big]\\ &=\sum_{y}\ketbra{y}\otimes \big[\big(\mbc{i,1}(A)\circ \mbc{i,1}(B)\big)+\mbc{i,0}(\iden)\big]~\text{(Using Lemma~\ref{lemma:cproduct_1})}\\
        &= \sum_{y}\ketbra{y}\otimes \big(\big[\mbc{i,1}(A) + \mbc{i,0}(\iden)\big]\circ \big[\mbc{i,1}(B)+\mbc{i,0}(\iden)\big]\big)~\text{(Using Lemma~\ref{lemma:cproduct_1})}\\
        &= \Big\{\sum_{y}\ketbra{y}\otimes \big[\mbc{i,1}(A) + \mbc{i,0}(\iden)\big]\Big\}\circ \Big\{\sum_{y}\ketbra{y}\otimes \big[\mbc{i,1}(B) + \mbc{i,0}(\iden)\big]\Big\}~\text{(Using Lemma~\ref{lemma:product_distributive})}
    \end{align*}
\end{proof}
\section{Proof of Theorem~\ref{thm:erroed-oracle-amplification-our}}
\label{appendix:biased-aa}

\biasedAAthm*

    Set $k = O(\frac{2p}{(p-1/2)^2}\log(1/\delta'))$ for a $\delta'=\lambda^4\delta^2$ and construct $\hata_{A,\borcl_p,k}$.
    This $\hata$ (dropping the subscripts) behaves as \[\hata\ket{0^n}\ket{0^k}\ket{0} = \sum_x \alpha_x \ket{x} \big(\eta_{x,0}\ket{\phi_{x,0}}\ket{0}+\eta_{x,1}\ket{\phi_{x,1}}\ket{1}\big) = \ket{\Psi}\] where $|\eta_{x,f(x)}|^2 \ge 1-\delta'$ for any $x$;  here $f(x)$ indicates the ``goodness'' of $x$.
    
    Now, two cases can happen.
    
    \textbf{Case (i): } Let $f(x)=0$ for all $x$. We analyse the situation that the algorithm does not output ``No Solution'', in other words, $R_{maj}$ was observed as $\ket{1}$. 
    
    
    Now, the output state after $\hata$ would be such that $|\eta_{x,1}|^2\le \delta'$ for all $x$.
    So, the probability of measuring $\ket{1}$ as output is $\sum_x |\alpha_x \eta_{x,1}|^2 \le \delta'\sum_x |\alpha_x|^2 = \delta'$. Can such a state, with final qubit as $\ket{1}$, appear with overwhelming probability after $O(1/\sqrt{\lambda})$ iterations of amplitude amplification? We argue not by lower bounding the number of iterations needed to boost the probability of such a state to almost certainty.
    
    Let $\theta$ be the angle made by the superposition of those states of $\ket{\Psi}$ whose last qubit is in $\ket{1}$. Then, we have $\sin^2(\theta) \le \delta' = \lambda^4\delta^2$.
    
    For any state $\ket{\chi}$ if the probability of obtaining a good state is $\sin^2(\theta) = \delta_1$ and if we would like to boost the probability to $\delta_2$, then it easy to show that the number of iterations needed in the amplitude amplification algorithm is $j = \Big\lceil\frac{1}{2}\big(\sin^{-1}(\sqrt{\delta_2})/\sin^{-1}(\sqrt{\delta_1})\big)-\frac{1}{2}\Big\rceil > \frac{1}{4}\big(\sin^{-1}(\sqrt{\delta_2})/\sin^{-1}(\sqrt{\delta_1})\big)$.
    Since $\theta < \sin^{-1}(\theta)$, we have \[j>\frac{1}{4}\big(\sin^{-1}(\sqrt{\delta_2})/\sin^{-1}(\sqrt{\delta_1})\big) > \frac{1}{4}\big(\sqrt{\delta_2}/\sin^{-1}(\sqrt{\delta_1})\big).\]
    In our case, we have $\delta_1 = \lambda^4\delta^2$ and $\delta_2=\delta$. So, the number of iterations required is \[j > \frac{1}{4}\Big(\sqrt{\delta}/\sin^{-1}(\sqrt{\lambda^4\delta^2})\Big) = \frac{1}{4}\Big(\sqrt{\delta}/\sin^{-1}(\lambda^2\delta)\Big).\] 
    For any $\beta\le 0.75$, it is easy to see that $\sin^{-1}(\beta)< \sqrt{\beta}$. Since, we set $\delta < 0.5$ and since $\lambda\le 1$, we have $\lambda^2\delta\le 0.5 < 0.75$. Hence we have,
    \[j > \frac{1}{4}\Big(\sqrt{\delta}/\sin^{-1}(\lambda^2\delta)\Big) > \frac{1}{4}\Big(\sqrt{\delta}/\sqrt{\lambda^2\delta}\Big) = \frac{1}{4\lambda}.\]
    This says that the number of amplification iterations required for improving the probability of obtaining $\ket{1}$ from $\lambda^2\delta^2$ to $\delta$ is at least $1/4\lambda$. But since the maximum number of iterations performed in the amplification routine is $O(\frac{1}{\sqrt{\lambda}})$, the probability of obtaining $\ket{1}$ on measuring the last qubit of the state after amplitude amplification is at most $\delta$ (most likely quite less).
    
    \textbf{Case (ii): } Let $f(x)=1$ for some $x$. In this case, for all $x$ such that $f(x)=1$, we will have $|\eta_{x,1}|^2\ge 1-\delta'$. Then the probability of measuring the last qubit as $\ket{1}$ is at least $\sum_{x : f(x)=1}|\alpha_x \eta_{x,1}|^2 \ge \lambda(1-\delta') > \lambda/2$ (since $\delta < 0.5$).
    Now, using the fixed point amplitude amplification subroutine, in $O(\frac{1}{\sqrt{\lambda}})$ iterations, we obtain a final state post amplification such that with probability $1-\delta$ we obtain $\ket{1}$ on measuring the $R_{maj}$ register.

\newcommand{\goodg}{\mathscr{G}\xspace}

Let the post-measurement state, after observing $R_{maj}$ in the state $\ket{1}$, be denoted  $\ket{\psi_m}$. We want to clarify that it is not immediately obvious that we shall observe a good state on measuring the first register of $\ket{\psi_m}$ since the biased oracle also marks the bad states with some probability. This requires an a additional analysis.

\begin{claim}
    Let $\ket{\psi_m}$ be the post-measurement state obtained on measuring the last qubit as $\ket{1}$. If the set of good state $\mathscr{G}=\{x : f(x)=1\}$ is non-empty, then the probability of obtaining some $x\in\mathscr{G}$ on measuring the first register of $\ket{\psi_m}$ is at least $3/4$.
\end{claim}
\begin{proof}
    The state just before amplification can be given as 
    \[\ket{\Psi} = \sum_x \alpha_x \ket{x} \big(\eta_{x,0}\ket{\phi_{x,0}}\ket{0} + \eta_{x,1}\ket{\phi_{x,1}}\ket{1}\big)\]
    where $|\eta_{x,f(x)}|^2 \ge 1-\delta'$ for any $x$.
    The probability of obtaining some good state on the condition that the $R_{maj}$ qubit is $\ket{1}$ is
    \begin{align*}
        Pr\Big[\ket{g}_{R_1} \Big| \ket{1}_{R_{maj}}\Big] &= \frac{Pr\Big[\ket{g}_{R_1}\ket{1}_{R_{maj}}\Big]}{Pr\Big[\ket{1}_{R_{maj}}\Big]}
        &= \frac{\sum_{x\in\goodg} |\alpha_x\eta_{x,1}|^2}{\sum_{x\in\goodg} |\alpha_x\eta_{x,1}|^2 + \sum_{x\notin\goodg} |\alpha_x\eta_{x,1}|^2}
        &= \frac{P_g}{P_g + P_b}~\text{ (say) }
    \end{align*}
    where by $Pr\Big[\ket{g}_{R_1}\Big]$ we denote the probability of obtaining some good state in $R_1$.
    We know that 
    \[P_b = \sum_{x\notin\goodg} |\alpha_x\eta_{x,1}|^2 = \sum_{x\notin\goodg} |\alpha_x|^2|\eta_{x,1}|^2 \le \delta'\sum_{x\notin\goodg} |\alpha_x|^2 \le \delta'.\]
    So, we have
    \[Pr\Big[\ket{g}_{R_1} \Big| \ket{1}_{R_{maj}}\Big] = \frac{P_g}{P_g + P_b} \ge \frac{P_g}{P_g+\delta'} = \frac{1}{1+(\delta'/P_g)}.\]
    Now,
    \begin{align*}
        \frac{\delta'}{P_g} &= \frac{\delta'}{\sum_{x\in\goodg}|\alpha_x|^2|\eta_{x,1}|^2}\\
        &\le \frac{\delta'}{(1-\delta')\sum_{x\in\goodg}|\alpha_x|^2}~~\big(\text{ Since } |\eta_{x,1}|^2\ge 1-\delta'~\text{ for } x\in\goodg\big)\\
        &\le \frac{\delta'}{(1-\delta')\lambda}~~\big(\text{ Since } |\alpha_{x}|^2\ge \lambda~\text{ for } x\in\goodg\big)\\
        &=\frac{\lambda^4\delta^2}{(1-\lambda^4\delta^2)\lambda} = \frac{\lambda^3\delta^2}{1-\lambda^4\delta^2} \le \frac{\delta^2}{1-\delta^2}\\
        &\le \frac{1/4}{1-(1/4)} = 1/3~~\big(\text{ Since } \delta \le 1/2\big).
    \end{align*}
    Using this, we get
    \[Pr\Big[\ket{g}_{R_1} \Big| \ket{1}_{R_{maj}}\Big] \ge \frac{1}{1+(\delta'/P_g)} \ge \frac{1}{1+(1/3)} = \frac{3}{4}.\]
    This gives us that if $R_{maj}$ was measured as $\ket{1}$ then on measuring $R_1$, with probability at least 3/4, we obtain $\ket{x}$ as measurement outcome for which $f(x)=1$.
\end{proof}

\section{Lower bound for Non-linearity Estimation}
\label{appendix:non-lin-lb}

Recall that the non-linearity of a Boolean function $f:\{0,1\}^n\xrightarrow{} \{0,1\}$ is defined as 
$$\eta(f) = \frac{1}{2}-\frac{1}{2}\fmax$$
where $\fmax = \max_x |\hat{f(x)}|$ and $\hat{f}(x)$ is the Walsh coefficient of $f$ at the point $x$.
We define a decision problem, namely the $\fmax$ decision problem, as follows: given a Boolean function $f:\{0,1\}^n\xrightarrow{} \{0,1\}$, a threshold $\tau$ and a parameter $\lambda$, decide if $\fmax \ge \tau$ or if $\fmax < \tau-\lambda$ given the promise that one of the two cases is true.
It is quite straight forward that the $\fmax$ problem can be directly reduced to the problem of non-linearity estimation.
So, to show a lower bound for the non-linearity estimation problem, we show a reduction from the \countdec problem to the $\fmax$ decision problem.
First consider the following lemma which will help prove the required reduction.

\begin{lemma}
    \label{lemma:coundec-nonlin-lb}
    The query complexity of any quantum algorithm that solves \countdec($N/4, N/4-\Delta$) is $\Omega(N/\Delta)$ for any $0 < \Delta \le N/5$.
\end{lemma}
\begin{proof}
    Using Theorem~\ref{thm:nayak_wu}, we obtain that the query complexity is
    \begin{align*}
        Q_{\countdec} &= \Omega\Bigg(\sqrt{\frac{N}{\Delta}} + \frac{\sqrt{\bigg(\frac{N}{4}-\Delta\bigg)\bigg(N-\bigg(\frac{N}{4}-\Delta\bigg)\bigg)}}{\Delta}\Bigg)\\
        &= \Omega\Bigg(\sqrt{\frac{N}{\Delta}} + \frac{\sqrt{\frac{3}{16}N^2 + \Delta N - \Delta^2}}{\Delta}\Bigg)\\
        &= \Omega(N/\Delta).
    \end{align*}
\end{proof}

\nonlinlb*

\begin{proof}
    For simplicity let $N$ be some power of $2$. Consider the \countdec($N/4, N/4-\Delta$) problem for some $0 < \Delta \le N/5$.
    The task is to decide if the Hamming weight of the given string $x$ is $N/4$ or $N/4-\Delta$.
    
    Now, for a given string $x$, construct a Boolean function $f^{(x)}:\{0,1\}^n\xrightarrow{}\{0,1\}$ such that $f^{(x)}(i)=x_i$ where $n=\log(N)$.
    We show that the problem of deciding if the Hamming weight of $x$ is $N/4$ or $N/4-\Delta$ can be solved by deciding if $\fmax^{(x)}$ is $\frac{1}{2}$ or $\frac{1}{2}+\frac{2\Delta}{N}$.
    
    Let $y$ be any string of Hamming weight $N/4$. Let $f^{(y)}$ be the Boolean function constructed using $y$.
    We know that the Walsh coefficient of function $f$ at $a$ is defined as 
    \begin{align*}
        \hat{f}(a) &= \frac{1}{2^n}\sum_{x\in \{0,1\}^n}(-1)^{f(x)\oplus a\cdot x}\\
        &= \frac{1}{2^n}\Big[\big|\{x\in \{0,1\}^n : f(x)=a\cdot x\}\big| - \big|\{x\in \{0,1\}^n : f(x)\neq a\cdot x\}\big|\Big].
    \end{align*}
    Intuitively, $|\hat{f}(a)|$ gives the difference in the fraction of inputs $x$ for which the function $f$ matches with the linear function $a\cdot x$ and the fraction of inputs for which the function does not match with $a\cdot x$.
    From this we can compute the Walsh coefficient of $f^{(y)}$ at $0^n$ to be 
    $$\hat{f}^{(y)}(0^n) = \frac{1}{2^n}\Bigg(\frac{3N}{4}-\frac{N}{4}\Bigg) = \frac{1}{2}.$$
    Now, let $a\neq 0^n$ be some $n$-bit string.
    Then, $a\cdot x$ is a linear function with equal number of $0$'s and $1$'s in its output.
    See that, for any Boolean function whose Hamming weight\footnote{By the Hamming weight of a Boolean function $f$, we mean the number of $1$'s in the output of $f$.} is $N/4$, the maximum number of inputs such that $f(x)=a\cdot x$ is bounded above by $3N/4$ where $N/2$ inputs has to be such that $a\cdot x = 0 = f(x)$ and $N/4$ inputs has to be such that $a\cdot x = 1 = f(x)$.
    So, we have the Walsh coefficient of $f^{(y)}$ at any $a\neq 0^n$ as
    $$\hat{f}^{(y)}(a) \le \frac{1}{2^n}\Bigg(\frac{3N}{4}-\frac{N}{4}\Bigg) = \frac{1}{2}.$$
    So, we have that $\fmax^{(y)} = \frac{1}{2}$ and it occurs at $0^n$.

    Next, let $z$ be a string of Hamming weight $N/4-\Delta$ and let $f^{(z)}$ be the Boolean function constructed from $z$.
    For $f^{(z)}$, we have that $$\hat{f}^{(z)}(0^n) = \frac{1}{2^n}\bigg[\Big(\frac{3N}{4}+\Delta\Big) - \Big(\frac{N}{4} - \Delta\Big)\bigg] = \frac{1}{2}+\frac{2\Delta}{2^n}.$$
    Again, for any Boolean function $f$ of Hamming weight $N/4-\Delta$, the maximum number of inputs such that $f(x)=a\cdot x$ is $\frac{3N}{4}-\Delta$ where $N/2$ inputs has to be such that $a\cdot x = 0 = f(x)$ and $N/4-\Delta$ inputs has to be such that $a\cdot x = 1 = f(x)$.
    So, we get that the Walsh coefficient of $f^{(z)}$ at any $a\neq 0^n$ is 
    $$\hat{f}^{(z)}(a) \le \frac{1}{2^n}\bigg[\Big(\frac{3N}{4}-\Delta\Big) - \Big(\frac{N}{4} + \Delta\Big)\bigg] = \frac{1}{2}-\frac{2\Delta}{2^n}.$$
    Thus, we get that $\fmax^{(z)} = \frac{1}{2}+\frac{2\Delta}{2^n}$ and it occurs at $0^n$.

    Consequently, any algorithm that solves the $\fmax$ decision problem for the parameters $\tau=\frac{1}{2}+\frac{2\Delta}{N}$ and $\lambda=\frac{\Delta}{N}$ can solve the \countdec($\frac{N}{4},\frac{N}{4}-\Delta$) problem without any query overhead.
    Now, using Lemma~\ref{lemma:coundec-nonlin-lb}, we get that any quantum algorithm that solves the $\fmax$ decision problem is $\Omega(\frac{N}{\Delta}) = \Omega(\frac{1}{\lambda})$.
\end{proof}

\end{document}